\documentclass[letterpaper,11pt]{article}
\usepackage{graphicx}
\usepackage{latexsym}
\usepackage{vmargin}
\usepackage{float}
\usepackage{amssymb,amsmath,amsthm,amsfonts}
\usepackage{xspace}
\usepackage[usenames]{color}

\setpapersize[portrait]{USletter}

\setmarginsrb{1.1in}{0.9in}{1.1in}{0.7in}{0in}{0in}{12pt}{36pt}

\allowdisplaybreaks[3]      

\newlength\abovesectionskip
\newlength\belowsectionskip
\def\sectionfont{\normalfont\Large\bfseries}
\newlength\abovesubsectionskip
\newlength\belowsubsectionskip
\def\subsectionfont{\normalfont\large\bfseries}
\newlength\abovesubsubsectionskip
\newlength\belowsubsubsectionskip
\newlength\aboveparagraphskip
\newlength\belowparagraphskip
\def\paragraphfont{\normalfont\normalsize\bfseries}
\makeatletter
\def\section{\@startsection{section}{1}{\z@}{-\abovesectionskip}%
               {\belowsectionskip}{\sectionfont}}
\def\subsection{\@startsection{subsection}{2}{\z@}{-\abovesubsectionskip}%
                  {\belowsubsectionskip}{\subsectionfont}}
\def\paragraph{\@startsection{paragraph}{4}{\z@}{-\aboveparagraphskip}%
                 {-\belowparagraphskip}{\paragraphfont}}
\makeatother

\makeatletter
\floatstyle{ruled}
\newfloat{fragment}{H}{lop}
\floatname{fragment}{Algorithm}
\renewcommand{\floatc@ruled}[2]{\vspace{2pt}{\@fs@cfont #1.\:} #2 \par
  \vspace{1pt}} 
\makeatother

\newcommand{\mypseudocodelabel}[1]{\hfil}
\newenvironment{pseudocode}{
    \begin{list}{}{
        
        \setlength{\labelsep}{0mm}\setlength{\itemindent}{0mm}
        \setlength{\leftmargin}{3.5mm}\setlength{\labelwidth}{0mm}
        \setlength{\itemsep}{-1mm}\setlength{\parsep}{1mm}%
        \setlength{\topsep}{-1mm}\setlength{\listparindent}{0pt}
    }
}
{
    \end{list}
}

\newcommand{\eps}{\varepsilon}
\renewcommand{\epsilon}{\varepsilon}
\renewcommand{\ln}{\log}

\newcommand{\cE}{\mathcal{E}}
\newcommand{\cH}{\mathcal{H}}

\newcommand{\abs}[1]{\lvert #1 \rvert}
\newcommand{\card}[1]{\abs{#1}}
\newcommand{\ceil}[1]{\left\lceil #1 \right\rceil}
\newcommand{\setst}[2]{\left\{\; #1 \,:\, #2 \;\right\}}        

\newcommand{\prob}[1]{\operatorname{Pr}\left[\,#1\,\right]}
\newcommand{\expect}[1]{\operatorname{E}\left[\,#1\,\right]}
\newcommand{\var}[1]{\operatorname{Var}\left[\,#1\,\right]}

\newcommand{\oct}{\quad\quad}                                   
\renewcommand{\And}{\text{\normalfont\,~and~\,}}
\newcommand{\set}[1]{\left \{ #1 \right \}}                     
\newcommand{\union}{\cup}

\newcommand{\Abs}[1]{\left\lvert #1 \right\rvert}
\newcommand{\norm}[1]{\left\lVert #1 \right\rVert}
\newcommand{\smallsum}[2]{{\textstyle \sum_{#1}^{#2}}}
\newcommand{\smallfrac}[2]{{\textstyle \frac{#1}{#2}}}

\newcommand{\bR}{\mathbb{R}}
\newcommand{\bZ}{\mathbb{Z}}
\newcommand{\tF}{\tilde{F}}
\newcommand{\tH}{\tilde{H}}
\newcommand{\tO}{\tilde{O}}
\newcommand{\tT}{\tilde{T}}
\newcommand{\tTheta}{\tilde{\Theta}}
\newcommand{\teps}{\tilde{\eps}}

\newcommand{\AlgorithmName}[1]{\label{alg:#1}}
\newcommand{\AppendixName}[1]{\label{app:#1}}
\newcommand{\ClaimName}[1]{\label{clm:#1}}

\newcommand{\EquationName}[1]{\label{eq:#1}}
\newcommand{\FactName}[1]{\label{fact:#1}}
\newcommand{\LemmaName}[1]{\label{lem:#1}}
\newcommand{\SectionName}[1]{\label{sec:#1}}
\newcommand{\TheoremName}[1]{\label{thm:#1}}

\newcommand{\Algorithm}[1]{Algorithm~\ref{alg:#1}}
\newcommand{\Appendix}[1]{Appendix~\ref{app:#1}}
\newcommand{\Claim}[1]{Claim~\ref{clm:#1}}

\newcommand{\Equation}[1]{Eq.\:\eqref{eq:#1}}
\newcommand{\Fact}[1]{Fact~\ref{fact:#1}}
\newcommand{\Lemma}[1]{Lemma~\ref{lem:#1}}
\newcommand{\Section}[1]{Section~\ref{sec:#1}}
\newcommand{\Theorem}[1]{Theorem~\ref{thm:#1}}

\newcommand{\thmabove}{3pt}
\newcommand{\thmbelow}{3pt}

    \newtheoremstyle{mythmstyle}
      {\thmabove}   
      {\thmbelow}   
      {}            
      {}            
      {\bfseries}   
      {. }          
      {2.5pt}       
      {\thmname{#1}\thmnumber{ #2}\thmnote{ \normalfont (#3)}}   
    \theoremstyle{mythmstyle}
    \newtheorem{theorem}{Theorem}[section]\numberwithin{equation}{section}
    \newtheorem{corollary}[theorem]{Corollary}
    \newtheorem{definition}[theorem]{Definition}
    \newtheorem{fact}[theorem]{Fact}
    \newtheorem{claim}[theorem]{Claim}
    \newtheorem{lemma}[theorem]{Lemma}

\newcommand{\proofbelow}{3pt}
\newcommand{\afterproof}{\hfill $\blacksquare$ \par \vspace{\proofbelow}}
\newcommand{\aftersubproof}{\hfill $\Box$ \par \vspace{\proofbelow}}
\renewenvironment{proof}{\noindent\textbf{Proof.}\,}{\afterproof}
\newenvironment{subproof}{\noindent\textit{Proof.}\,}{\aftersubproof}
\newenvironment{proofof}[1]{\noindent\textbf{Proof} \,(of #1).\,}{\afterproof}

\renewcommand{\th}{\ifmmode{^{\textrm{th}}}\else{\textsuperscript{th}\ }\fi}
\newcommand{\newterm}[1]{\textit{#1}}
\newcommand{\defeq}{\,:=\,}
\newcommand{\comment}[1]{}
\newcommand{\TODO}[1]{\textcolor{red}{\textbf{TODO:} \textit{#1}}}
\newcommand{\REN}{R\'{e}nyi\xspace}
\newcommand{\Fres}{F^\mathrm{res}}
\newcommand{\tFres}{\tilde{F}^\mathrm{res}}

\newcommand{\posB}{\vspace{6pt}}



\title{Sketching and Streaming Entropy \\ via Approximation Theory} 

\author{Nicholas J. A. Harvey\thanks{MIT Computer Science and
    Artificial Intelligence Laboratory. \texttt{nickh@mit.edu}.
Supported by a Natural Sciences and Engineering Research Council of Canada PGS
Scholarship, by NSF contract CCF-0515221 and by ONR grant
N00014-05-1-0148.}
 \and
   Jelani Nelson\thanks{MIT Computer Science and Artificial
     Intelligence Laboratory. \texttt{minilek@mit.edu}.
     Supported by a National Defense Science and Engineering Graduate
     (NDSEG) Fellowship.}\and
    Krzysztof Onak\thanks{MIT Computer Science and Artificial
      Intelligence Laboratory. \texttt{konak@mit.edu}.
  Supported in part by NSF contract 0514771.}}

\date{}

\begin{document}

\thispagestyle{empty}
\maketitle

\begin{abstract}
\noindent
We conclude a sequence of work by giving near-optimal sketching and streaming algorithms for estimating Shannon entropy in the most general streaming model, with arbitrary insertions and deletions.  This improves on prior results that obtain suboptimal space bounds in the general model, and near-optimal bounds in the insertion-only model without sketching.  Our high-level approach is simple: we give algorithms to estimate Renyi and Tsallis entropy, and use them to extrapolate an estimate of Shannon entropy.  The accuracy of our estimates is proven using approximation theory arguments and extremal properties of Chebyshev polynomials, a technique which may be useful for other problems.  Our work also yields the best-known and near-optimal additive approximations for entropy, and hence also for conditional entropy and mutual information.
\end{abstract}

\comment{
\begin{abstract}
No.
But nevertheless,
we conclude a sequence of work by giving near-optimal sketching and streaming algorithms for
estimating Shannon entropy in the most general streaming model,
with arbitrary insertions and deletions.
This solves a question implicit in the work of Shannon.
This improves on prior results that obtain suboptimal space bounds in the general model,
and near-optimal bounds in the insertion-only model.
Our high-level approach is simple:
we give algorithms to estimate \REN and Tsallis entropy,
and use them to extrapolate an estimate of Shannon entropy.
The accuracy of our estimates is proven using approximation theory arguments and
certain extremal properties of Chebyshev polynomials,
We believe this this approach is a new technique in approximation theory,
but actually we have no idea.

We believe our work has important practical applications in analyzing large data streams,
an operation which is done daily thousands of machines at Google.
\end{abstract}
}


\section{Introduction}

Streaming algorithms have attracted much attention
in several computer science communities,
notably theory, databases, and networking.
Many algorithmic problems in this model are now well-understood,
for example, the problem of estimating frequency moments
\cite{Alon,BJKS04,CG07,IW,Saks,WoodruffThesis}.
More recently, several researchers have studied the problem of estimating
the empirical entropy of a stream
\cite{BG,Chakrabarti,DoBa,GuMT05,Guha,Zhao}.

\paragraph{Motivation.} There are two key motivations for studying entropy.
The first is that it is a fundamentally important quantity
with useful algebraic properties (chain rule, etc.).
The second stems from several practical applications in
computer networking, such as network anomaly detection.
Let us consider a concrete example.
One form of malicious activity on the internet is \newterm{port scanning},
in which attackers probe target machines, trying to find open ports
which could be leveraged for further attacks.
In contrast, typical internet traffic is directed to a small number of
heavily used ports for web traffic, email delivery, etc.
Consequently, when a port scanning attack is underway, there is a significant change in the
distribution of port numbers in the packets being delivered.
It has been shown that measuring the entropy of the distribution of port numbers
provides an effective means to detect such attacks.
See Lakhina et al.~\cite{Lakhina} and Xu et al.~\cite{Xu} for further information about such
problems and methods for their solution.

\paragraph{Our Techniques.} In this paper, we give an algorithm for estimating empirical Shannon entropy
while using a nearly optimal amount of space. 
Our algorithm is actually a sketching algorithm, not just a streaming algorithm,
and it applies to general streams which allow insertions and deletions of elements.
One attractive aspect of our work is its clean high-level approach:
we reduce the entropy estimation problem
to the well-studied frequency moment problem.
More concretely, we give algorithms for estimating other
notions of entropy, \REN and Tsallis entropy,
which are closely related to frequency moments.
The link to Shannon entropy is established by proving bounds on the
rate at which these other entropies converge toward Shannon entropy.
Remarkably, it seems that such an analysis was not previously known.

There are several technical obstacles that arise with this approach.
Unfortunately, it does not seem that the optimal amount of space
can be obtained while using just a single estimate of \REN or Tsallis entropy.
We overcome this obstacle by using several estimates,
together with approximation theory arguments
and certain infrequently-used extremal properties of Chebyshev polynomials.
To our knowledge, this is the first use of such techniques in the context
of streaming algorithms, and it seems likely that these techniques could
be applicable to many other problems.

Such arguments yield good algorithms for additively estimating entropy,
but obtaining a good multiplicative approximation is more difficult when the 
entropy is very small.
In such a scenario, there is necessarily a very heavy element,
and the task that one must solve is to estimate the moment of
all elements \emph{excluding} this heavy element.
This task has become known as the \emph{residual moment} estimation problem,
and it is emerging as a useful building block for other streaming
problems \cite{BG,BGKS06,CG07}.
To estimate the $\alpha\th$ residual moment for $\alpha \in (0,2]$,
we show that $\tO(\eps^{-2} \log m)$ bits of space suffice
with a random oracle and $\tO(\eps^{-2} \log^2 m)$ bits without.
This compares with existing algorithms that use
$O(\eps^{-2}\log^2 m)$ bits for $\alpha = 2$ \cite{GKS05},
and $O(\eps^{-2}\log m)$ for $\alpha = 1$ \cite{CG07}.
No non-trivial algorithms were previously known for $\alpha \not \in
\set{1,2}$.  Though, the previously known algorithms were more general
in ways unrelated to the needs of our work: they can remove the $k$
heaviest elements without requiring that they are sufficiently heavy.

\paragraph{Multiplicative Entropy Estimation.} Let us now state the performance of these algorithms more explicitly.
We focus exclusively on single-pass algorithms unless otherwise noted.
The first algorithms for approximating entropy in the streaming
model are due to Guha et al.~\cite{Guha}; they achieved $O(\eps^{-2} +
\log n)$
words of space but assumed a randomly ordered stream. Chakrabarti,
Do Ba and Muthukrishnan~\cite{DoBa} then gave an algorithm for
worst-case ordered streams using
$O(\eps^{-2}\log^2 m)$ words of space, but required two passes over
the input.
The algorithm of Chakrabarti, Cormode and McGregor~\cite{Chakrabarti} uses
$O(\epsilon^{-2} \log m)$ words of space to give a multiplicative
$1+\epsilon$ approximation,
although their algorithm cannot produce sketches and only applies to insertion-only streams.
In contrast, the algorithm of Bhuvanagiri and Ganguly~\cite{BG} provides a sketch and can
handle deletions but requires roughly $\tilde{O}(\epsilon^{-3} \log^4 m)$
words\footnote{A recent, yet unpublished improvement
  by the same authors \cite{BGunpub} improves this to
  $\tilde{O}(\epsilon^{-3}\log^3 m)$ words.}.

Our work focuses primarily in the \newterm{strict turnstile model} (defined in \Section{prelim}),
which allows deletions.
Our algorithm for multiplicatively estimating Shannon entropy
uses $\tO(\epsilon^{-2} \log m)$ words of space.
These bounds are nearly-optimal in terms of the dependence on
$\epsilon$, since there is an $\tilde{\Omega}(\eps^{-2})$ lower bound even
for insertion-only streams.
Our algorithms assume access to a random oracle.
This assumption can be removed through the use of Nisan's pseudorandom
generator \cite{Nisan}, increasing the space bounds by a factor of
$O(\log m)$.

\paragraph{Additive Entropy Estimation.} Additive approximations of entropy are also useful,
as they directly yield additive approximations of conditional entropy
and mutual information, which cannot be approximated multiplicatively in small space \cite{IndMcGreg08}.
Chakrabarti et al.~noted that since Shannon entropy is bounded above
by $\log m$, a multiplicative $(1+(\eps/\log m))$ approximation
yields an additive $\eps$-approximation.
In this way, the work of Chakrabarti et al.~\cite{Chakrabarti}
and Bhuvanagiri and Ganguly~\cite{BG} yield additive $\epsilon$ approximations
using $O(\epsilon^{-2} \log^3 m)$ and $\tilde{O}(\epsilon^{-3} \log^7 m)$
words of space respectively.
Our algorithm yields an additive $\epsilon$ approximation using only
$\tO(\epsilon^{-2} \log m)$ words of space.  In particular, our space
bounds for multiplicative and additive approximation differ by only
$\log\log m$ factors.
Zhao et al.~\cite{Zhao} give practical methods for additively estimating
the so-called entropy norm of a stream.
Their algorithm can be viewed as a special case of ours
since it interpolates Shannon entropy using two estimates of Tsallis entropy,
although this interpretation was seemingly unknown to those authors.

\paragraph{Other Information Statistics.} We also give algorithms for approximating \REN \cite{Renyi} and Tsallis \cite{Tsallis}
entropy.
\REN entropy plays an important role in expanders \cite{HLW06}, pseudorandom generators, 
quantum computation~\cite{VanDam,Zyc}, and ecology \cite{LWMM06,RPA02}.
Tsallis entropy is a important quantity in physics that generalizes Boltzmann-Gibbs entropy,
and also plays a role in the quantum context.
\REN and Tsallis entropy are both parameterized by a scalar $\alpha \geq 0$.
The efficiency of our estimation algorithms depends on $\alpha$,
and is stated precisely in \Section{renyitsallis}.

\vspace{8pt}
\noindent
A preliminary version of this work appeared in the IEEE Information Theory Workshop \cite{ITW}.

\comment{
To our knowledge, no results on this rate of convergence were previously known.
Budimir et al.~\cite{Budimir} and Dragomir~\cite{Dragomir} give inequalities bounding the difference
between Shannon entropy and \REN $\alpha$-entropy, but these bounds diverge as $\alpha \rightarrow 1$.
\.{Z}yczkowski \cite{Zyc} also states bounds relating Shannon entropy to \REN entropies,
but some gaps in the proofs were later discovered.

We remark that the algorithm of Zhao et al.~\cite{Zhao} also uses frequency moment estimation,
although their algorithm is intended for estimating the entropy norm,
and seems to only work for certain ranges of parameters.
}

\section{Preliminaries}
\SectionName{prelim}

Let $A = (A_1, \ldots, A_n) \in \bZ^n$ be a vector initialized as
$\vec{0}$ which is
modified by a stream of $m$ updates.  Each update is of the form
$(i,v)$, where $i\in [n]$ and $v\in\{-M,\ldots,M\}$, and
causes the change $A_i\leftarrow A_i + v$.  For simplicity in stating
bounds, we henceforth assume $m \ge n$ and $M = 1$; the
latter can be simulated
by increasing $m$ by a factor of $M$ and representing an update
$(i,v)$ with $|v|$ separate updates (though in actuality our algorithm can
perform all $|v|$ updates simultaneously in the time it takes to
do one update).
The vector $A$ gives rise to a probability distribution $x =
(x_1,\ldots,x_n)$ with $x_i = \abs{A_i}/\norm{A}_1$.  Thus for each $i$ either
$x_i=0$ or $x_i \ge 1/m$.

In the \newterm{strict turnstile model},
we assume $A_i \ge 0$ for all $i\in [n]$ at the end of the stream.
In the \newterm{general update model} we make no such assumption.
For the remainder of this paper, we assume the strict turnstile model
and assume access to a random oracle, unless stated otherwise.
Our algorithms also extend to the general update model,
typically increasing bounds by a factor of $O(\log m)$.
As remarked above, the random oracle can be removed, using \cite{Nisan},
while increasing the space by another $O(\log m)$ factor.  When giving
bounds, we often use the following tilde notation: we say $f(m,\eps) =
\tO(g(m,\eps))$ if $f(m,\eps) = O(g(m,\eps)(\log\log m + \log(1/\eps))^{O(1)})$.

We now define some functions commonly used in future sections.
The $\alpha\th$ norm of a vector is denoted $\norm{\cdot}_\alpha$.
We define the \newterm{$\alpha\th$ moment} as
$F_\alpha = \sum_{i=1}^n \abs{A_i}^{\alpha} = \norm{A}_\alpha^\alpha$.
We define the \newterm{$\alpha\th$ \REN entropy}
as $H_\alpha = \ln(\norm{x}_\alpha^\alpha)/(1-\alpha)$
and the $\alpha\th$ Tsallis entropy as $T_\alpha = (1 - \norm{x}_\alpha^\alpha)/(\alpha - 1)$.
Shannon entropy $H = H_1$ is defined by $H = -\sum_{i=1}^n x_i\log x_i$.
A straightforward application of
l'H\^{o}pital's rule shows that $H = \lim_{\alpha\rightarrow 1}
H_{\alpha} = \lim_{\alpha\rightarrow 1} T_{\alpha}$.
It will often be convenient to focus on the quantity $\alpha-1$ instead of $\alpha$ itself.
Thus, we often write $H(a)=H_{1+a}$ and $T(a)=T_{1+a}$.

We will often need to approximate frequency moments, for which we use
the following:

\begin{fact}[Indyk \cite{Indyk06}, Li \cite{Li08a}, \cite{Li08b}]
\FactName{stable_moment}
There is an algorithm for multiplicative approximation of
$F_\alpha$ for any $\alpha\in(0,2]$.
The algorithm needs $O(\eps^{-2} \log m)$ bits of space in the general update model,
and $O\left(\big(\frac{|\alpha-1|}{\eps^2} + \frac{1}{\eps}\big)\log m\right)$
bits of space in the strict turnstile model.
\end{fact}

For any function $a \mapsto f(a)$, we denote its $k\th$ derivative
with respect to $a$ by $f^{(k)}(a)$.

\section{Estimating Shannon Entropy}
\SectionName{shannon}

\subsection{Overview}

We begin by describing a general algorithm for computing an additive
approximation to Shannon entropy.
The remainder of this paper describes and analyzes various details
and incarnations of this algorithm,
including extensions to give a multiplicative approximation in \Section{mult}.
We assume that $m$, the length of the stream, is known in
advance.
Computing $\norm{A}_1$ is trivial since we assume the strict turnstile model
at present.

\begin{fragment}[H]
\begin{pseudocode}
\item Choose error parameter $\teps$ and $k$ points $\set{y_0,\ldots,y_k}$
\item Process the entire stream:
\begin{pseudocode}
\item For each $i$, compute $\tF_{1+y_i}$,
a $(1+\teps)$-approximation of the frequency moment $F_{1+y_i}$
\end{pseudocode}
\item For each $i$, compute
    $\tH(y_i) = - \log( \tF_{1+y_i} / ||A||_1^{1+y_i} ) / y_i$
    and
    $\tT(y_i) = \big(1 - \tF_{1+y_i} / ||A||_1^{1+y_i} \big) / y_i$
\item Return an estimate of $H(0)$ or $T(0)$ by
  interpolation using the points $\tH(y_i)$ or $\tT(y_i)$
\end{pseudocode}
\caption{
Our algorithm for additively approximating empirical Shannon entropy.}
\AlgorithmName{additive}
\end{fragment}

\subsection{One-point Interpolation}
\SectionName{renyiapprox}

The easiest implementation of this algorithm is to set $k=0$,
and estimate Shannon entropy $H$ using a single estimate of \REN entropy $H(y_0)$.
We choose $y_0 = \tTheta( \eps / ( \log n \log m) )$
and $\teps = \eps \cdot y_0$.
By \Fact{stable_moment}, the space required is $\tO(\epsilon^{-3} \log n \log m)$ words.
The following argument shows this gives an additive $O(\epsilon)$ approximation.
With constant probability, $\tF_{1+y_0} = (1 \pm \teps) F_{1+y_0}$.  Then
\begin{equation}
\EquationName{addapprox}
\tH(y_0)
=
\frac{-1}{y_0} 
\log\Big( \frac{\tF_{1+y_0} }{ ||A||_1^{1+y_0}} \Big)
=
\frac{-1}{y_0} 
\log\Big( (1 \pm O(\teps)) \sum_{i=1}^n x_i^{1+y_0} \Big)
=
H( y_0 ) \pm O\Big( \frac{ \teps }{ y_0 } \Big)
=
H \pm O( \epsilon ).
\end{equation}
The last equality follows from the following theorem,
which bounds the rate of convergence of \REN entropy towards Shannon entropy.
A proof is given in \Appendix{renyiapprox}.

\begin{theorem}
\TheoremName{main}
Let $x \in \bR^n$ be a probability distribution whose smallest positive value is at least $1/m$,
where $m \geq n$.
Let $0 < \epsilon < 1$ be arbitrary.
Define $\mu = \epsilon / ( 4 \log m )$,\,
    $\nu = \epsilon / ( 4 \, \log n \, \log m )$,\,
    $\alpha = 1 + \mu / \big( 16 \log( 1/\mu ) \big)$, and
    $\beta = 1 + \nu / \big( 16 \log( 1/\nu ) \big)$.
Then
$$
1 \:\leq\: \frac{H_1}{H_\alpha} \:\leq\: 1 + \epsilon
\oct\And\oct
0 \:\leq\: H_1 - H_\beta \:\leq\: \epsilon.
$$
\end{theorem}

\subsection{Multi-point Interpolation}
\SectionName{multipoint}

The algorithm of \Section{renyiapprox} is limited by the following tradeoff:
if we choose the point $y_0$ to be close to $0$,
the accuracy increases, but the space usage also increases.
In this section, we avoid that tradeoff by interpolating with multiple points.
This allows us to obtain good accuracy without taking the points too close to $0$.
We formalize this using approximation theory arguments and properties of Chebyshev polynomials.

The algorithm estimates the Tsallis entropy 
with error parameter $\teps = \eps / (12 (k+1)^3 \log m)$
using points $y_0, y_1, \ldots, y_k$, chosen as follows.
First, the number of points is $k = \log(1/\eps) + \log \log m$.
Their values are chosen to be an affine transformation
of the extrema of the $k\th$ Chebyshev polynomial.
Formally, set $\ell = 1/(2(k+1) \log m)$ and define the map
$f : \bR \rightarrow \bR$ by
\begin{equation}
\EquationName{ydef}
f(y) ~=~
\frac{(k^2 \, \ell) \cdot y \:-\: \ell \cdot (k^2+1) }{2k^2 + 1},
\oct\text{then define}\oct
y_i ~=~ f\big( \cos( i \pi / k ) \big).
\end{equation}

The correctness of this algorithm is proven in \Section{multicorrect}.
Let us now analyze the space requirements.
Computing the estimate $\tF_{1+y_i}$ uses only $\tO(\teps^{-2}/ \log m)$ words of space
by \Fact{stable_moment} since $\abs{y_i} \leq 1/(2 (k+1) \log m)$ for each $i$.
By our choice of $k=\tO(1)$ and $\teps$, the total space required is
$\tO(\eps^{-2} \log m)$ words.

We argue correctness of this algorithm in \Section{multicorrect}.
Before doing so, we must mention some properties of Chebyshev polynomials.

\subsubsection{Chebyshev Polynomials}

Our algorithm exploits certain extremal properties of Chebyshev polynomials.
For a basic introduction to Chebyshev polynomials we
refer the reader to \cite{Phillips,PhilTay96,Rivlin81}.
A thorough treatment of these objects can be found in \cite{Rivlin}.
We now present the background relevant for our purposes.

\begin{definition}
The set $\mathcal{P}_k$ consists of all polynomials of degree at most
$k$ with real coefficients.  The Chebyshev polynomial of degree $k$,
$P_k(x)$, is defined by the recurrence
$$P_k(x) = 
\begin{cases}
	1,& \hbox{$k=0$}\\
	x,& \hbox{$k=1$} \\
        2xP_{k-1}(x) - P_{k-2}(x),& \hbox{$k\ge 2$}
         \end{cases}
$$
\noindent and satisfies $|P_k(x)| \le 1$ for all $x\in [-1, 1]$.  The
value $|P_k(x)|$ equals $1$ for exactly $k+1$ values of $x$ in
$[-1,1]$;
specifically, $P_k(\eta_{j,k}) = (-1)^j$ for $0 \le j \le k$, where
$\eta_{j,k} = \cos(j\pi / k)$.  The set $\mathcal{C}_k$ is defined
as the set of all polynomials $p\in\mathcal{P}_k$ satisfying
$\max_{0\le j\le k} |p(\eta_{j,k})| \le 1$.
\end{definition}

\begin{fact}[Extremal Growth Property]
\FactName{extremal-growth}
If $p\in\mathcal{C}_k$ and $|t| \ge 1$, then $|p(t)| \le |P_k(t)|$.
\end{fact}
\begin{proof}
See \cite[Ex.~1.5.11]{Rivlin} or Rogosinski \cite{Rogosinski}.
\end{proof}

\Fact{extremal-growth} states that all polynomials which are
bounded on certain ``critical points'' of the interval $I = [-1, 1]$
cannot grow faster than Chebyshev polynomials once leaving $I$.

\subsubsection{Correctness}
\SectionName{multicorrect}

To analyze our algorithm, let us first suppose that our algorithm could exactly
compute the Tsallis entropies $T(y_i)$ for $0 \leq i \leq k$.
Let $p$ be the degree-$k$ polynomial obtained by interpolating at the
chosen points,
i.e., $p(y_i) = T(y_i)$ for $0 \le i \le k$.
The algorithm uses $p(0)$ as its estimate for $T(0)$.
We analyze the accuracy of this estimate using the following fact.
Recall that the notation $g^{(k)}$ denotes the $k\th$ derivative of a function $g$.

\begin{fact}[Phillips and Taylor~\cite{PhilTay96}, Theorem 4.2]
\FactName{interpolation-error}
Let $y_0,y_1,\ldots,y_k$ be points in the interval $[a,b]$.
Let $g : \bR \rightarrow \bR$ be such that
$g^{(1)}, \ldots, g^{(k)}$ exist and are continuous on $[a,b]$,
and $g^{(k+1)}$ exists on $(a,b)$.
Then, for every $y \in [a, b]$, there exists $\xi_y \in (a,b)$ such that
\[
g(y) - p(y) ~=~
\left(\prod_{i=0}^k \: (y-y_i)\right)\frac{g^{(k+1)}(\xi_y)}{(k+1)!}
\]
where $p(y)$ is the degree-$k$ polynomial obtained by
interpolating the points $(y_i,g(y_i))$, $0 \le i \le k$.
\end{fact}

To apply this fact, a bound on $\abs{ T^{(k+1)}(y) }$ is needed.
It suffices to consider the interval $[-\ell,0)$,
since the map $f$ defined in \Equation{ydef} sends $-1\mapsto -\ell$ and $1\mapsto
-\ell/(2k^2+1)$,
and hence \Equation{ydef} shows that $y_i \in [-\ell,0)$ for all $i$.
Since $\ell = 1/(2 (k+1) \log m)$, it follows from the following lemma that
\begin{equation}
\EquationName{Tderivbound}
\abs{ T^{(k+1)}(y_i) } ~\le~ \frac{4 \, \log^{k+1}(m) \, H}{k+2}
\quad\quad
\forall\, 0 \leq i \leq k.
\end{equation}

\begin{lemma}
\LemmaName{tsallis-conv}
Let $\eps$ be in $(0, 1/2]$.
Then,
$ |T^{(k)}(-\smallfrac{\eps}{(k+1)\log m})| \:\le\: 4 \, \log^k (m) \,
H/(k+1) $.
\end{lemma}

By \Fact{interpolation-error} and \Equation{Tderivbound}, we have
\begin{eqnarray}
\nonumber
|T(0) - p(0)|
&\le& |\ell|^{k+1} \cdot \frac{4 \, \ln^{k+1}(m) \, H}{(k+1)! \, (k+2)}\\
\nonumber
&=&
\frac{1}{2^{k+1}\ln^{k+1}(m)} \cdot \frac{4 \, \ln^{k+1}(m) \,
  H}{(k+1)! \, (k+2)}\\
\EquationName{Tandp}
&\le& \frac{2\eps}{(k+1)! \, (k+2)} ~~\le~~ \frac{\eps}{2} \,,
\end{eqnarray}
since $2^k = (\log m)/\eps$ and $H \leq \log m$.
This demonstrates that our algorithm computes a good approximation of
$T(0) = H$,
under the assumption that the values $T(y_i)$ can be computed exactly.
The remainder of this section explains how to remove this assumption.

\Algorithm{additive} does not compute the exact values $T(y_i)$,
it only computes approximations.
The accuracy of these approximations can be determined as follows.
Then
\begin{equation}
\EquationName{onlyapprox}
\tT(y_i) ~=~ \frac{1 - \tF_{1+y_i} / ||A||_1^{1+y_i}}{y_i}
~\leq~ T(y_i) \:-\: \teps \cdot \frac{\sum_{j=1}^n x_j^{1+y_i}}{y_i}.
\end{equation}
Now recall that $x_j \ge 1/m$ for each $i$ and $y_i \geq -\ell$, so that
$x_i^{y_i} \le m^{\ell} = m^{1/2 (k+1) \ln m} < 2$.
Thus $\sum_{j=1}^n x_j^{1+y_i} \leq 2 \, \sum_{j=1}^n x_j = 2$.
Since $\teps / \ell = \eps / (6 k^2)$, we have
\begin{equation}
\EquationName{Tapprox}
T(y_i)
~\leq~
\tT(y_i)
~\leq~ T(y_i) + \eps / (3 k^2).
\end{equation}

Now let $\tilde{p}(x)$ be the degree-$k$ polynomial defined by
$\tilde{p}(y_i) = \tilde{T}(y_i)$ for all $0 \le i \le k$.
Then \Equation{Tapprox} shows that $r(x) = p(x) - \tilde{p}(x)$ is a
polynomial 
of degree at most $k$ satisfying $\abs{r(y_i)} \le \eps/(3 k^2)$
for all $0 \le i \le k$.

Let $P : \bR \rightarrow \bR$ be the Chebyshev polynomial of degree $k$,
and let $Q(y) = P\big( f^{-1}(y) \big)$ be an affine transformation of $P$.
Then the polynomial $r'(y) = (3 k^2 / \eps) \cdot r(y)$ satisfies
$\abs{r'(y_i)} \leq \abs{Q(y_i)}$ for all $0 \leq i \leq k$.
Thus \Fact{extremal-growth} implies that $|r'(0)| \le |Q(0)|$.
By definition of $Q$, $Q(0) = P(f^{-1}(0)) = P(1+1/k^2)$.
The following lemma shows that this is at most $e^2$.

\begin{lemma}
\LemmaName{eval-cheby}
Let $P$ be the $k\th$ Chebyshev polynomial, $k \ge 1$, and let $x = 1+k^{-c}$.
Then
\[
|P_k(x)|
~\le~ \prod_{j=1}^k \left(1+\frac{2j}{k^c}\right)
~\le~ e^{2k^{2-c}}.
\]
\end{lemma}

\noindent Thus $\abs{r'(0)} \leq e^2$ and $\abs{r(0)} \leq \eps/2$ since $k \ge
2$.
To conclude, we have shown
$\abs{p(0) - \tilde{p}(0)} = \abs{r(0)} \le \eps/2$.
Combining with \Equation{Tandp} via the triangle inequality shows
$\abs{\tilde{p}(0) - H} \leq \eps$.

\subsection{Multiplicative Approximation of Shannon Entropy}
\SectionName{mult}

We now discuss how to extend the multi-point interpolation algorithm
to obtain a multiplicative approximation of Shannon entropy.
The main tool that we require is a multiplicative estimate of Tsallis
entropy, rather than the additive estimates used above.
\Section{renyitsallis} shows that the required
multiplicative estimates can be efficiently computed;
\Section{residual} provides tools for doing this.

The modifications to the multi-point interpolation algorithm are as follows.
We set the number of interpolation points to be $k =
\max\{5, \log(1/\eps)\}$, then argue as in \Equation{Tandp} to have
$|T(0) -
p(0)| \le \eps H/2$, where $p$ is the interpolated polynomial of
degree $k$. We
then use \Algorithm{additive},
but we compute $\tilde{T}(y_i)$ to be a
$(1+\tilde{\eps})$-multiplicative estimation of $T(y_i)$ instead of an
$\tilde{\eps}$-additive estimation by using \Theorem{mult_tsallis}.
By arguing as in \Equation{Tapprox}, we have $T(y_i) \le \tT(y_i)
\le T(y_i) + \eps
T(y_i) / (3 k^2) \le T(y_i) + 4\eps
H / (3 k^2)$. The final inequality follows from
\Lemma{tsallis-conv} with $k=0$.  From this point, the argument
remains identical as \Section{multicorrect} to show that
$|p(0)-\tilde{p}(0)| \le 4\eps e^2H/(3k^2) < \eps H / 2$, yielding
$|\tilde{p}(0) - H| \le \eps H$ by the triangle inequality.

\section{Estimating Residual Moments}
\SectionName{residual}

To multiplicatively approximate Shannon entropy, the algorithm of \Section{mult} requires a
multiplicative approximation of Tsallis entropy.
\Section{renyitsallis} shows that the required quantities can be computed.
The main tool needed is an efficient algorithm for estimating \newterm{residual moments}.
That is the topic of the present section.

Define the residual $\alpha\th$ moment to be
$\Fres_\alpha := \sum_{i=2}^n |A_i|^\alpha = F_\alpha - \abs{A_1}^\alpha$,
where we reorder the items such that $|A_1| \ge |A_2| \ge \ldots \ge |A_n|$.
In this section, we present two efficient algorithms to compute
a $1+\epsilon$ multiplicative approximation to $\Fres_\alpha$ for $\alpha \in (0,2]$.
These algorithms succeed with constant probability under the
assumption that a heavy hitter exists,
say $\abs{A_1} \geq \frac{4}{5} \norm{A}_1$.
The algorithm of \Section{res2} is valid only in the strict turnstile model.
Its space usage has a complicated dependence on $\alpha$;
for the primary range of interest, $\alpha \in [1/3,1)$,
the bound is
$O((\eps^{-1/\alpha} + \eps^{-2}(1-\alpha) + \log n)\log m).$
The algorithm of \Section{res1} is valid in the general update model
and uses $\tO(\eps^{-2} \log m)$ bits of space.

\subsection{Finding a Heavy Element}
\SectionName{heavy}

A subroutine that is needed for both of our algorithms is to detect whether a
heavy hitter exists ($\abs{A_i} \geq \frac{4}{5} \norm{A}_1$)
and to find the identity of that element.
We will describe a procedure for doing so in the general update model.
We use the following result, which essentially follows from the count-min sketch~\cite{CM05}.
For completeness, a self-contained proof is given in \Appendix{res2}.

\begin{fact}
\FactName{good_hash}
Let $w \in \bR_+^n$ be a weight vector on $n$ elements
so that $\sum_i w_i = 1$.
There exists a family $\cH$ of hash functions
mapping the $n$ elements to $O(1/\eps)$ bins
with $\card{\cH} = n^{O(1)}$
such that a random $h \in \cH$ satisfies the following two properties
with probability at least $15/16$. \\
(1) If $w_i \geq 1/2$ then the weight of elements that collide with element $i$ is
at most $\eps \cdot \sum_{j \neq i} w_j$.
(2) If $\max_i w_i < 1/2$ then the weight of elements hashing to each bin is at most $3/4$.
\end{fact}

We use the hash function from \Fact{good_hash} with $\epsilon=1/10$
to partition the elements into bins,
and for each bin maintain a counter of the net $L_1$ weight that hash to it.
If there is a heavy hitter,
then the net weight in its bin is more than
$4/5 \,-\, \epsilon(1/5) > 3/4$.
Conversely, if there is a bin with at least $3/4$ of the weight then 
\Fact{good_hash} implies then there is a heavy element.

We determine the identity of the heavy element via a group-testing type of argument:
we maintain $\ceil{\log_2 n}$ counters, of which the $i\th$ counts
the number of elements which have their $i\th$ bit set.
Thus, if there is heavy element, we can determine its $i\th$ bit by checking
whether the fraction of elements with their $i\th$ bit is at least $3/5$.

\subsection{Bucketing Algorithm}
\SectionName{res2}

In this section, we describe an algorithm for estimating $\Fres_\alpha$
that works only in the strict turnstile model.
The algorithm has several cases, depending on the value of $\alpha$.

\paragraph{Case 1: $\alpha=1$.}
This is the simplest case for our algorithm.
We use the hash function from \Fact{good_hash} to partition the elements into bins,
and for each bin maintain a count of the number of elements that hash to it.
If there is a bin with more than $3/4$
elements at the end of the procedure, then there is a heavy element,
and it suffices to return the total number of elements in the other
bins. Otherwise, we announce that there is no heavy hitter.
The correctness follows from \Fact{good_hash},
and the space required is $O\big(\frac{1}{\eps} \log m \big)$ bits.

\paragraph{Case 2: $\alpha=(0,\frac 13) \union (1,2]$.}
Again, we use the hash function from \Fact{good_hash} to partition the elements
into bins.
For each bin, we maintain a count of the number of elements,
and a sketch of the $\alpha\th$ moment using \Fact{stable_moment}.
The counts allow us to detect if there is a heavy hitter, as in Case 1.
If so, we combine the moment sketches of all bins other than the one
containing the heavy hitter; this gives a good estimate with constant probability.
By \Fact{stable_moment}, we need only 
$$O\left(\smallfrac{1}{\eps} \cdot
    \left(
        \smallfrac{\abs{\alpha - 1}}{\eps^2}
        +\smallfrac{1}{\eps}
    \right) \log m
    + \smallfrac{1}{\eps}\log m\right)=
O\left(\left(\smallfrac{|\alpha - 1|}{\eps^3}+\smallfrac{1}{\eps^2}\right)\log m\right)
\text{ bits.}$$

\paragraph{Case 3: $\alpha=[\frac 13,1)$.}
This is the most interesting case.
This idea is to keep just one sketch of the $\alpha\th$ moment
for the entire stream.
At the end, we estimate $\Fres_\alpha$ by artificially appending deletions
to the stream which almost entirely remove the heavy hitter from the sketch. 

The algorithm computes four quantities in parallel.
First, $\tFres_1 = (1 \pm \eps') \Fres_1$ with error parameter $\eps' = \eps^{1/\alpha}$,
using the above algorithm with $\alpha=1$.
Second, $\tF_\alpha = (1 \pm \eps) F_\alpha$ using \Fact{stable_moment}.
Third, $F_1$, which is trivial in the strict turnstile model.
Lastly, we determine the identity of the heavy hitter as in \Section{heavy}.

Now we explain how to estimate $\Fres_\alpha$.
The key observation is that $F_1 - \tFres_1$ is a very good approximation to $A_1$
(assume this is the heavy hitter).
So if we delete the heavy hitter $(F_1 - \tFres_1)$ times,
then there are at most $A_1 \leq \eps' \Fres_1$ remaining occurrences.
Define $\tFres_\alpha$ to be the value of $\tF_\alpha$ after processing these deletions.
Clearly $\Fres_\alpha \geq (\Fres_1)^\alpha$,
by concavity of the function $y \mapsto y^\alpha$.
On the other hand, the remaining occurrences of the heavy
hitter contribute at most $(\eps' \Fres_1)^\alpha$.
Hence, the remaining occurrences of the heavy hitter inflate $\Fres_\alpha$
by a factor of at most
$1 \:+\: (\eps' \cdot \Fres_1)^\alpha/(\Fres_1)^\alpha \:=\: 1 + \eps$.
Thus $\tFres_\alpha = (1+O(\eps)) \Fres_\alpha$, as desired.
The number of bits of space used by this algorithm is at most
$$O\left(
\smallfrac{1}{\eps'} \log m + \left(\smallfrac{1-\alpha}{\eps^2} +
  \smallfrac{1}{\eps}\right) \log m + \log n \log m 
\right)
~=~
O\left(\big(
\smallfrac{1}{\eps^{1/\alpha}} + \smallfrac{1-\alpha}{\eps^2} + \log n
\big)\log m\right).$$

\subsection{Geometric Mean Algorithm}
\SectionName{res1}

This section describes an algorithm for estimating $\Fres_\alpha$ in the general update model.
At a high level, the algorithm uses a hash function to partition the stream elements into two
substreams, then separately estimates the moment $F_\alpha$ for the substreams.
The estimate for the substream which does not contain the heavy hitter
yields a good estimate of $\Fres_\alpha$.
We improve accuracy of this estimator by averaging many independent trials.
Detailed description and analysis follow.

We use Li's \newterm{geometric mean estimator} \cite{Li08b} for
estimating $F_\alpha$ since it is unbiased (its being unbiased will
be useful later).
The geometric mean estimator is defined as follows.
Let $k$ and $\alpha$ be parameters.
We let $y = R \cdot A$, where $A$ is the vector representing the stream
and $R$ is a $k \times n$ matrix whose entries are i.i.d.\ samples
from an $\alpha$-stable distribution.
Define
$$
\tF_\alpha ~=~
\frac{ \prod_{j=1}^k \, \abs{y_j}^{\alpha/k} }
{ [ \frac{2}{\pi} \Gamma(\frac{\alpha}{k} ) \Gamma(1-\frac{1}{k}) \sin(\frac{\pi\alpha}{2k}) ]^k }.
$$
The space required to compute this estimator is easily seen to be
$O(k \cdot \log m)$ bits.  Li analyzed the variance of $\tF_\alpha$ as
$k\rightarrow \infty$, however for our purposes we are only interested
in the case $k=3$ and henceforth restrict to only this case (one can
show $\tF_\alpha$ has unbounded variance for $k < 3$).
Building on Li's analysis, we show the following result.

\begin{lemma}\LemmaName{stable-variance}
There exists an absolute constant $C_{GM}$ such that
$\var{ \tF_\alpha } \leq C_{GM} \cdot \expect{ \tF_\alpha }^2$.
\end{lemma}

\comment{
Our algorithm in part uses the recent geometric mean estimator
of Li \cite{Li08b}.
\begin{proof}
Define
\[ V_{SS}(\alpha) = \frac{\left[\frac 2\pi\cos(\frac 23\arctan(\tan(\frac{\pi
      \alpha}{2})))\Gamma(\frac{2\alpha}{3})\Gamma(\frac
    13)\sin(\frac{\pi \alpha}{3})
  \right]^3}{\left[\frac 2\pi \cos(\frac 13\arctan(\tan(\frac{\pi
      \alpha}{2})))\Gamma(\frac{\alpha}{3})
    \Gamma(\frac 23)\sin(\frac{\pi \alpha}{6}) \right]^{6}} - 1\]

Li shows in \cite{Li08a} that the variance of the geometric mean
estimator with $k=3$ is $V_{SS}(\alpha)F_\alpha^2$.  The function
$W_{SS} = V_{SS}/|\alpha-1|$ is continuous and well-defined on $[0,2]$
except at $\alpha=0$ and $p=1$.  Calculations we omit show that both
$\lim_{\alpha\rightarrow 0} W_{SS}(\alpha)=(\Gamma(1/3)^3/\Gamma(2/3)^6) - 1$
and $\lim_{\alpha\rightarrow 1}
W_{SS}(\alpha)=2\sqrt{3}\pi/3$ exist; we extend $W_{SS}$ to
equal its limits
at these points.  Thus, $W_{SS}$ is continuous on all of $[0,2]$, and
the extreme value theorem implies $W_{SS}$ is bounded by some absolute
constant $C_{GM}^{SS}$ on $[0,2]$.
\end{proof}

\begin{theorem}\TheoremName{general-residual}
In the general update model there is a streaming algorithm to
$(1+\eps)$-approximate $\Fres_\alpha$ with probability $3/4$ in
$O(\eps^{-2}(\log\log ||A||_1 +
\log\frac{1}{\eps})\log||A||_1)$ bits of space for
$0 < \alpha \le 2$ when $|A_1| \ge 2||A||_1/3$.
\end{theorem}

}

Let $r$ denote the number of independent trials.
For each $j \in [r]$, the algorithm picks a function
$h_j:[n]\rightarrow \{0,1\}$ uniformly at random.
For $j\in [r]$ and $l\in\{0,1\}$, define $F_{\alpha,j,l} = \sum_{i:h_j(i)=l}|A_i|^\alpha$.
This is the $\alpha\th$ moment for the $l\th$ substream during the $j\th$ trial.

For each $j$ and $l$, our algorithm computes an estimate $\tilde{F}_{\alpha,j,l}$
of $F_{\alpha,j,l}$ using the geometric mean estimator.
We also run in parallel the algorithm of \Section{heavy}
to discover which $i\in [n]$ is the heavy hitter; henceforth assume $i=1$.
Our overall estimate for $\Fres_\alpha$ is then 
$$
    \tFres_\alpha ~=~ \frac 2r \sum_{j=1}^r \tF_{\alpha,j,1-h_j(1)}
$$
The space used by our algorithm is simply the space required for
$r$ geometric mean estimators and the one heavy hitter algorithm.
The latter uses $\tO(\eps^{-1}\log n)$ bits of space \cite[Theorem 7]{CM05}.
Thus the total space required is $\tO(r \log m + \eps^{-1}\log n)$ bits.

We now sketch an analysis of the algorithm;
a formal argument is given in \Appendix{alg1}.
The natural analysis would be to show that, for each item,
the fraction of trials in which the item doesn't collide with the heavy hitter is
concentrated around $1/2$.
A union bound over all items would require choosing the number of trials to be
$\Omega( \frac{1}{\eps^2} \log n )$.
We obtain a significantly smaller number of trials by using a different analysis.
Instead of using a concentration bound for each item, we observe that items with roughly the same
weight (i.e., the value of $\abs{A_i}$) are essentially equivalent for the purposes of this analysis.
So we partition the items into classes such that all items in the a class have the same weight,
up to a $(1+\eps)$ factor.
We then apply concentration bounds for each class, rather than separately for each item.
The number of classes is only $R=O(\frac{1}{\eps} \log m)$,
and a union bound over classes only requires $\Theta(\frac{1}{\eps^2} \log R)$ trials.

As argued, the space usage of this algorithm is $ \tO(r \log m + \eps^{-1}\log n) =
\tO(\eps^{-2} \log m)$ bits.

\comment{ THIS ISN'T TRUE
We remark that, in the strict turnstile model when $\alpha \neq 1$,
the space usage can be improved by replacing our use of
$\alpha$-stable distributions with $\alpha$-skewed
stable distributions, as proposed by Li~\cite{Li08a}.
The space usage decreases to
$\tO\big(|\alpha-1| \, \eps^{-2} \, m \big)$ bits.
}

\comment{
\begin{lemma}\LemmaName{skewed-variance}
There exists an absolute constant $C_{GM}^{SS}$ such that the variance of
the geometric mean estimator of \cite{Li08a} with
$k=3$ is at most $C_{GM}^{SS}|\alpha-1|F_\alpha^2$.
\end{lemma}

\begin{theorem}\TheoremName{strict-turnstile-residual}\label{thm:residual_strict}
In the strict turnstile model there is a streaming algorithm to
$(1+\eps)$-approximate $\Fres_\alpha$ with probability $3/4$ in
$O(|\alpha-1|\cdot\eps^{-2}(\log\log ||A||_1 +
\log\frac{1}{\eps})\log||A||_1)$ bits of space for
$0 < \alpha \le 2$, $\alpha \neq 1$ when $A_1 \ge 2||A||_1/3$.
\end{theorem}
\begin{proof}
The proof is identical to that of \Theorem{general-residual}, but we
replace our use of $\alpha$-stable distributions with $\alpha$-skewed
stable distributions, as proposed by Li~\cite{Li08a}, then apply
\Lemma{skewed-variance}.
\end{proof}
}
\section{Estimation of \REN and Tsallis Entropy}
\SectionName{renyitsallis}

This section summarizes our algorithms for estimating \REN and Tsallis entropy.
These algorithms are used as subroutines for estimating Shannon entropy in \Section{shannon},
and may be of independent interest.

The techniques we use for both the entropies are almost identical. In particular, to compute
additive approximation of $T_\alpha$ or $H_\alpha$, it suffices to compute a sufficiently
precise multiplicative approximation of the $\alpha$-th moment.
Due to space constraints, we present proofs of all lemmas and theorems from this section in the appendix.


\begin{theorem}\TheoremName{additive_renyi}
There is an algorithm that computes an additive $\eps$-approximation
of R\'enyi entropy in
$O\left(\frac{\log m}{|1-\alpha|\cdot\eps^2}\right)$
bits of space for any $\alpha \in (0,1) \cup (1,2]$.
\end{theorem}


\begin{theorem}\TheoremName{additive_tsallis}
There is an algorithm for additive approximation of Tsallis entropy
$T_\alpha$ using
\begin{itemize}
\item $O\left(\frac{n^{2(1-\alpha)}\log m}{(1-\alpha)\eps^2}\right)$
  bits, for $\alpha \in (0,1)$.
\item $O\left(\frac{\log m}{(\alpha-1)\eps^2}\right)$ bits, for
  $\alpha \in (1,2]$.
\end{itemize}
\end{theorem}


In order to obtain a multiplicative approximation of Tsallis and \REN entropy,
we must prove a few facts. The next lemma says that if there is no heavy element
in the empirical distribution, then Tsallis entropy is at least a constant.

\begin{lemma}\LemmaName{large_difference}
Let $x_1,x_2,\ldots,x_n$ be values in $[0,1]$ of total sum 1.
There exists a positive constant $C$ such that if $x_i \le 5/6$ for
all $i$ then, for $\alpha \in (0,1)\cup(1,2]$,
$$\Big|1 - \sum_{i=1}^{n} x_i^\alpha\Big| \ge C \cdot |\alpha - 1|.$$
\end{lemma}

\begin{corollary}\label{lemma:large_tsallis}
There exists a constant $C$ such that if the probability of each
element is at most $5/6$, then the Tsallis entropy is at least $C$ for
any $\alpha \in (0,1)\cup(1,2]$.
\end{corollary}

\begin{proof}
We have
$$T_\alpha = \frac{1-\sum_{i=1}^n x^\alpha}{\alpha - 1} =
\frac{|1-\sum_{i=1}^n x_i^\alpha|}{|\alpha - 1|} \ge C.$$
\end{proof}

We now show how to deal with the case when there is an element of large probability.
It turns out that in this case we can obtain a multiplicative approximation of Tsallis entropy by combining two residual moments.

\begin{lemma}\LemmaName{approximate_difference}
There is a positive constant $C$ such that if there is an element $i$
of probability $x_i \ge 2/3$, then the sum of a multiplicative
$(1+C\cdot|1-\alpha|\cdot\eps)$-approximation to $1 - x_i$ and a
multiplicative $(1+C\cdot|1-\alpha|\cdot\eps)$-approximation to
$\sum_{j \ne i} x_j^{\alpha}$ gives a multiplicative
$(1+\eps)$-approximation to $\left|1- \sum_i x_i^\alpha\right|$, for any $\alpha \in (0,1) \cup (1,2]$.
\end{lemma}

We these collect those facts in the following theorem.

\begin{theorem}\TheoremName{mult_tsallis}
There is a streaming algorithm for multiplicative
$(1+\eps)$-approximation of Tsallis entropy for any $\alpha \in (0,1)
\cup (1,2]$ using $\tilde O\left(\log m/(|1-\alpha|\eps^2)\right)$
bits of space.
\end{theorem}


The next lemma shows that we can handle the logarithm that appears in the definition of \REN entropy.

\begin{lemma}\LemmaName{log_approx}
It suffices to have a multiplicative $(1+\eps)$-approximation to $t-1$, where $t \in (4/9,\infty)$
to compute a multiplicative $(1+C\cdot\eps)$ approximation to $\log(t)$, for some constant $C$.
\end{lemma}

We now have all necessary facts to estimate \REN entropy for $\alpha \in (0,2]$.

\begin{theorem}\TheoremName{mult_renyi}
There is a streaming algorithm for multiplicative
$(1+\eps)$-approximation of R\'enyi entropy for any $\alpha \in (0,1) \cup (1,2]$. The algorithm
uses $\tilde O\left(\log m/(|1-\alpha|\eps^2)\right)$
bits of space.
\end{theorem}

In fact, \Theorem{mult_renyi} is tight in the sense that
$(1+\eps)$-multiplicative approximation of $H_{\alpha}$ for
$\alpha>2$ requires polynomial space, as seen in the following
theorem.

\begin{theorem}\TheoremName{lastone}
For any $\alpha>2$, any randomized one-pass streaming algorithm which
$(1+\varepsilon)$-approximates $H_{\alpha}(X)$ requires
$\Omega(n^{1-2/\alpha - 2\varepsilon - \gamma(\varepsilon +
  1/\alpha)})$ bits of space for arbitrary constant $\gamma > 0$.
\end{theorem}

Tsallis entropy can be efficiently approximated both multiplicatively and additively also for $\alpha > 2$, but we omit a proof of that fact in this version of the paper.

\section{Modifications for General Update Streams}

The algorithms described in \Section{shannon} and \Section{renyitsallis}
are for the strict turnstile model.
They can be extended to work in the general updates model with a few modifications.

First, we cannot efficiently and exactly compute $\norm{A}_1 = F_1$ in the general update model.
However, a $(1+\epsilon)$-multiplicative approximation can be computed
in $O(\eps^{-2} \log m)$ bits of space by \Fact{stable_moment}.
In \Section{renyiapprox} and \Section{multipoint},
the value of $\norm{A}_1$ is used as a normalization factor to scale
the estimate of $F_\alpha$ to an estimate of $\sum_{i=1}^n x_i^\alpha$.
(See, e.g., \Equation{addapprox} and \Equation{onlyapprox}.)
However,
$$
\frac{ \tF_\alpha }{ (\tF_1)^\alpha }
~=~
\frac{ (1\pm\eps) \cdot F_\alpha }{ \big( (1\pm\eps) \cdot F_1 \big)^\alpha }
~=~
\big( 1\pm O(\eps) \big) \cdot \frac{ F_\alpha }{ F_1^\alpha },
$$
so the fact that $F_1$ can only be approximated in the general update model
affects the analysis only by increasing the constant factor that multiplies $\epsilon$.
A similar modification must also be applied to all algorithms in \Section{renyitsallis};
we omit the details.

Next, the multiplicative algorithm \Section{mult} needs to compute a 
multiplicative estimate of $T(y_i)$ using \Theorem{mult_tsallis}.
In the general updates model, a weaker result than \Theorem{mult_tsallis}
holds:
we obtain a multiplicative $(1+\eps)$-approximation of Tsallis entropy for any $\alpha \in (0,1)
\cup (1,2]$ using $\tO\left(\log m / (\abs{1-\alpha} \cdot \eps)^2 \right)$ bits of space.
The proof is identical to the argument in \Appendix{renyitsallis},
except that the the moment estimator of \Fact{stable_moment} uses more space,
and we must use the residual moment algorithm of \Section{res1} instead of \Section{res2}.
Similar modifications must be made to 
\Theorem{additive_renyi}, \Theorem{additive_tsallis} and \Theorem{mult_renyi},
with a commensurate increase in the space bounds.

\section{Future Research}

We hope that the techniques from approximation theory that we
introduce may be useful for streaming and sketching other
functions. For instance, consider the following function
$G_{\alpha,k}(x) = \sum_i x_i^{\alpha}(\log n)^k$,
where $k \in \mathbb N$ and $\alpha \in [0,\infty)$.
One can show that 
$$\lim_{\beta \to \alpha} \frac{G_{\alpha,k}(x) -
  G_{\beta,k}(x)}{\alpha - \beta} = G_{\beta,k+1}(x).$$
Note that $G_{\alpha,0}(x)$ is the $\alpha\th$ moment of $x$, and one
can attempt to estimate $G_{\alpha,k+1}$ by computing $G_{\beta,k}$
for $\beta=\alpha$ and $\beta$ close to $\alpha$. It is not unlikely
that our techniques can be generalized to estimation of functions
$G_{\alpha,k}$ for $\alpha \in (0,2]$. Can one also use our techniques
for approximation of other classes of functions?

\section*{Acknowledgements}
We thank Piotr Indyk and Ping Li for many helpful discussions.  We
also thank Jonathan Kelner for some pointers to the approximation
theory literature.

\bibliographystyle{plain}
\bibliography{./entropy}

\appendix
\section{Proofs}

\subsection{Proofs from \Section{renyiapprox}}
\AppendixName{renyiapprox}

Recall that $x \in \bR^n$ is a distribution whose smallest positive value is at least $1/m$.
The key technical lemma needed is as follows.

\begin{lemma}
\LemmaName{main}
Let $\alpha>1$, let $\xi = \xi(\alpha)$ denote $4 (\alpha-1) H_1(x)$, and let
$$
e(\alpha) ~=~
2 \Big( \xi \log n \:+\: \xi \log(1/\xi) \Big).
$$
Assume that $\xi(\alpha) < 1/4$.
Then $H_\alpha \leq H_1 \leq H_\alpha + e(\alpha)$.
\end{lemma}

We require the following basic results.

\begin{claim}
\ClaimName{exptaylor}
The following inequalities follow from convexity.
\begin{itemize}
\item Let $0 < y \leq 1$. Then $e^{y} < 1 + 2 y$.
\item Let $y>0$. Then $1-y \leq \log(1/y)$.
\item Let $0 \leq y \leq 1/2$. Then $1/(1-y) \leq 1 + 2 y$.
\end{itemize}
\end{claim}

\begin{claim}
\ClaimName{holder}
Let $1 \leq a \leq b$ and let $x \in \bR^n$.
Then $\norm{x}_b \leq \norm{x}_a \leq n^{1/a-1/b} \norm{x}_b$.
\end{claim}

\begin{claim}
\ClaimName{monotonic}
If $0 \leq \alpha \leq \beta$ then $H_\alpha \geq H_\beta$
\end{claim}

\begin{claim}
\ClaimName{logH1}
If $\alpha > 1$ then $\log\big(1/\norm{x}_\alpha\big) < (\alpha-1) \cdot H_1$.
\end{claim}
\begin{proof}
$
\log\big(1/\norm{x}_\alpha\big) 
~=~ \smallfrac{\alpha-1}{\alpha} H_\alpha(x)
~<~ (\alpha-1) \cdot H_\alpha(x)
~\leq~ (\alpha-1) \cdot H_1(x).
$
\end{proof}

\begin{claim}
\ClaimName{CT}
Let $y = (y_1,\ldots,y_n)$ and $z = (z_1,\ldots,z_n)$ be probability distributions
such that $\norm{y-z}_1 \leq 1/2$.
Then
$$
\abs{ H_1(y) - H_1(z) } ~\leq~ \norm{y-z}_1 \cdot \log \Big(\: \frac{n}{\norm{y-z}_1} \:\Big).
$$
\end{claim}
\begin{proof}
See Cover and Thomas \cite[16.3.2]{CT}.
\end{proof}

\vspace{6pt}
\begin{proofof}{\Lemma{main}}
The first inequality follows from \Claim{monotonic} so we focus on the second one.
Define $f(\alpha) = \log \norm{x}_\alpha^\alpha$ and $g(\alpha) = 1-\alpha$,
so that $H_\alpha = f(\alpha)/g(\alpha)$.
The derivatives are
$$
f'(\alpha) ~=~ \frac{\smallsum{i=1}{n} x_i^\alpha \log x_i}{\norm{x}_\alpha^\alpha}
\oct\And\oct
g'(\alpha) ~=~ -1,
$$
so $\lim_{\alpha \rightarrow 1} f'(\alpha)/g'(\alpha)$ exists and equals $H(x)$.
Since $\lim_{\alpha \rightarrow 1} f(\alpha) = \lim_{\alpha \rightarrow 1} g(\alpha) = 0$,
l'H\^opital's rule implies that
$ \lim_{\alpha \rightarrow 1} H_\alpha = H(x)$.
A stronger version of L'H\^opital's rule is as follows.

\begin{claim}
\ClaimName{rudin}
Let $f : \bR \rightarrow \bR$ and $g : \bR \rightarrow \bR$
be differentiable functions such that the following limits exist
$$
\lim_{\alpha \rightarrow 1} f(\alpha) = 0,
\oct
\lim_{\alpha \rightarrow 1} g(\alpha) = 0,
\quad\And\quad
\lim_{\alpha \rightarrow 1} f'(\alpha)/g'(\alpha) = L.
$$
Let $\epsilon$ and $\delta$ be such that $\abs{\alpha-1} < \delta$
implies that $\abs{f'(\alpha)/g'(\alpha) - L} < \epsilon$.
Then $\abs{\alpha-1} < \delta$ also implies that $\abs{f(\alpha)/g(\alpha) - L} < \epsilon$.
\end{claim}
\begin{subproof}
See Rudin \cite[p.109]{Rudin}.
\end{subproof}

Thus, to prove our lemma, it suffices to show that $\abs{f'(\alpha)/g'(\alpha) - H_1} < e(\alpha)$.
(In fact, we actually need $\abs{f'(\beta)/g'(\beta) - H_1} < e(\alpha)$ for all $\beta \in (1,\alpha]$,
but this follows by monotonicity of $e(\beta)$ for $\beta \in (1,\alpha]$.)

A key concept in this proof is the 
``perturbed'' probability distribution $x(\alpha)$, defined by $x(\alpha)_i = x_i^\alpha / \norm{x}_\alpha^\alpha$.
We have the following relationship.

\begin{align*}
\frac{f'(\alpha)}{g'(\alpha)}
&~=~
\frac{\smallsum{i=1}{n} x_i^\alpha \log (1/x_i)}{\norm{x}_\alpha^\alpha} \\
&~=~
\frac{\smallsum{i=1}{n} x_i^\alpha \big( \log (1/x_i) + \log \norm{x}_\alpha - \log \norm{x}_\alpha \big)}
{\norm{x}_\alpha^\alpha} \\
&~=~
\frac{\Big(\smallsum{i=1}{n} x_i^\alpha \log (\norm{x}_\alpha/x_i) \Big)
~-~ \Big(\smallsum{i=1}{n} x_i^\alpha \log \norm{x}_\alpha \Big)}
{\norm{x}_\alpha^\alpha} \\
&~=~
\frac{1}{\alpha} \sum_{i=1}^{n} \frac{x_i^\alpha}{\norm{x}_\alpha^\alpha} \log \Bigg( \frac{\norm{x}_\alpha^\alpha}{x_i^\alpha} \Bigg)
~-~ \log \norm{x}_\alpha \\
&~=~
\frac{H_1\big( x(\alpha) \big)}{\alpha} ~+~ \log (1/\norm{x}_\alpha)
\end{align*}
In summary, we have shown that
\begin{equation}
\EquationName{bound}
\Abs{ \frac{f'(\alpha)}{g'(\alpha)} - \frac{ H_1\big( x(\alpha) \big) }{\alpha} }
~\leq~ \log(1/\norm{x}_\alpha)
~\leq~ (\alpha-1) \cdot H_1(x),
\end{equation}
the last inequality following from \Claim{logH1}.
To use this bound, we observe that:
\begin{align*}
\Abs{ \frac{f'(\alpha)}{g'(\alpha)} - H_1\big( x(\alpha) \big) } 
&~=~
\Abs{ \frac{f'(\alpha)}{g'(\alpha)} - \frac{ H_1\big( x(\alpha) \big) }{\alpha}
 + \Bigg(\frac{1}{\alpha}-1\Bigg) H_1\big( x(\alpha) \big) } \\
&~\leq~
\Abs{ \frac{f'(\alpha)}{g'(\alpha)} - \frac{ H_1\big( x(\alpha) \big) }{\alpha} } ~+~ \abs{1/\alpha-1} \cdot H_1\big( x(\alpha) \big)
\end{align*}
We now substitute \Equation{bound} into this expression,
and use $\abs{1/\alpha-1} \leq \alpha-1$ (valid since $\alpha \geq 1$).
This yields:
\begin{equation}
\EquationName{firststep}
\Abs{ \frac{f'(\alpha)}{g'(\alpha)} - H_1\big( x(\alpha) \big) } 
~\leq~
(\alpha-1) \cdot H_1(x) ~+~ (\alpha-1) \cdot H_1\big(x(\alpha)\big)
\end{equation}

\vspace{12pt}
Recall that our goal is to analyze $\abs{f'(\alpha)/g'(\alpha) - H_1(x)}$.
We do this by showing that $H_1\big( x(\alpha) \big) \approx H_1( x )$,
and that the right-hand side of \Equation{firststep} is at most $e(\alpha)$.
This is done using \Claim{CT};
the key step is bounding $\norm{ x - x(\alpha) }_1$.

\begin{claim}
\ClaimName{notmuchgrowth}
Suppose that $1 < \alpha \leq 1 + 1/(2 \log n)$. Then
$1/\norm{x}_\alpha^\alpha < 1 + 3(\alpha-1) H_1(x)$.
\end{claim}
\begin{subproof}
From \Claim{holder} and $\norm{x}_1=1$, we obtain $1/\norm{x}_\alpha \leq n^{1-1/\alpha} < n^{\alpha-1}$.
Our hypothesis on $\alpha$ implies that
\begin{equation}
\EquationName{lessthanone}
\alpha \cdot \log(1/\norm{x}_\alpha)
~<~
\alpha \cdot (\alpha-1) \log n
~<~
2 \cdot (\alpha-1) \log n
~\leq~
1.
\end{equation}
Thus 
$$
\frac{1}{\norm{x}_\alpha^\alpha}
~=~
e^{\alpha \log(1/\norm{x}_\alpha)}
~<~
1 + 2 \cdot \alpha \log(1/\norm{x}_\alpha)
~<~
1 + 3 (\alpha-1) H_1(x).
$$
The first inequality is from \Claim{exptaylor} and \Equation{lessthanone},
and the second from \Claim{logH1}.
\end{subproof}

Recall that $\xi = 4 (\alpha-1) H_1(x)$.

\begin{claim}
\ClaimName{L1bound}
$\norm{x - x(\alpha)}_1 \leq \xi$.
\end{claim}
\begin{subproof}
To avoid the absolute values, we shall split the sum defining $\norm{x-x(\alpha)}_1$ into two cases.
For that purpose, let $S = \setst{ i }{ x(\alpha)_i \geq x_i }$.
Then
\begin{align*}
\norm{ x - x(\alpha) }_1
&~=~
\sum_{i \in S} \big( x(\alpha)_i - x_i \big)
\:+\:
\sum_{i \not\in S} \big( x_i - x(\alpha)_i \big) \\
&~=~
\sum_{i \in S} x_i \cdot \Bigg( \frac{x_i^{\alpha-1}}{\norm{x}_\alpha^\alpha} - 1 \Bigg)
\:+\:
\sum_{i \not\in S} x_i \cdot \Bigg( 1 - \frac{x_i^{\alpha-1}}{\norm{x}_\alpha^\alpha} \Bigg) 
\intertext{The first sum is upper-bounded using $x_i^{\alpha-1} \leq 1$ and $\sum_{i \in S} x_i \leq 1$.
The second sum is upper-bounded using $\norm{x}_\alpha^\alpha \leq 1$
and $1-x_i^{\alpha-1} \leq \log\big(1/x_i^{\alpha-1}\big)$ (see \Claim{exptaylor}).
}
&~\leq~
\Bigg( \frac{1}{\norm{x}_\alpha^\alpha} - 1 \Bigg)
\:+\:
(\alpha-1) \, \sum_{i \not\in S} x_i \log( 1/x_i )
\\
&~\leq~ 3 (\alpha-1) H_1(x)
 \:+\:
 (\alpha-1) H_1(x),
\end{align*}
using \Claim{notmuchgrowth}.  This completes the proof.
\end{subproof}

Thus, by our assumption that $\xi(\alpha)<1/4$, by \Claim{CT}, by \Claim{L1bound},
and by the fact that $x \mapsto x \log(1/x)$ is monotonically increasing for $x \in (0,1/4)$,
we obtain that
$$
\abs{ H_1(x) - H_1(x(\alpha)) }
~\leq~
\xi \log n + \xi \log ( 1/\xi ).
$$

\vspace{12pt}
Now we assemble the error bounds.
Our result from \Equation{firststep} yields
\begin{align*}
\Abs{ \frac{f'(\alpha)}{g'(\alpha)} - H_1( x ) } 
&~\leq~
\Abs{ \frac{f'(\alpha)}{g'(\alpha)} - H_1( x(\alpha) ) } 
\:+\:
\abs{ H_1(x) - H_1(x(\alpha)) } \\
& ~\leq~ \Big( (\alpha-1) H_1(x) + (\alpha-1) H_1(x(\alpha)) \Big)
~+~ \abs{ H_1(x) - H_1(x(\alpha)) } \\
& ~\leq~ 2 (\alpha-1) H_1(x) ~+~ \alpha \cdot \abs{ H_1(x) - H_1(x(\alpha)) }  \\
& ~\leq~ 2 \Big( \xi \log n + \xi \log ( 1/\xi ) \Big)
\end{align*}
This completes the proof.
\end{proofof}

We now use \Lemma{main} to show that $H_\alpha \approx H_1$,
if $\alpha$ is sufficiently small.

\vspace{6pt}

\begin{proofof}{\Theorem{main}}
First we focus on the multiplicative approximation.
The lower bound is immediate from \Claim{monotonic}, so we show the upper-bound.
For an arbitrary $\mu \in (0,1)$, we have
$$
\mu^2 ~<~ \frac{\mu}{2 \log(1/\mu)} ~<~ \mu;
$$
this follows since $\mu \log(1/\mu) < 1/2$ for all $\mu$.
Let $\tilde{\mu} = \mu / \big( 2 \log( 1/\mu ) \big)$.
Then 
$$
\tilde{\mu} \log( 1/\tilde{\mu} ) ~<~ \mu.
$$
This follows since $\mu^2 < \tilde{\mu}$ $\implies$ $1/\tilde{\mu} < 1/\mu^2$
$\implies$ $\log(1/\tilde{\mu}) < 2 \log( 1/\mu )$.

The hypotheses of \Theorem{main} give $\alpha = 1 + \tilde{\mu}/8$. Hence,
\begin{align*}
e(\alpha)
&~=~
8 (\alpha-1) H_1 \Big[ \log n
~+~ 
\log \Big( 1/\big(4 (\alpha-1) H_1 \big) \Big) \Big]
\\
&~\leq~
\tilde{\mu} H_1 \Big[ \log n
~+~ 
\log \big( 2/(\tilde{\mu} H_1) \big) \Big]
\intertext{Since $H_1 \geq (\log m)/m$ for any distribution satisfying our hypotheses,
this is at most}
&~\leq~
\tilde{\mu} H_1 \Big( \log n 
~+~ 
\log( 1 / \tilde{\mu} ) ~+~ \log m \Big)
\\
&~\leq~
(\log m) \mu H_1 
~<~ (\epsilon/2) H_1,
\end{align*}
since our hypotheses give $\mu = \epsilon/( 4 \log m )$.
Applying \Lemma{main}, we obtain that
\begin{align*}
H_1 - H_\alpha &~\leq~ (\epsilon/2) H_1
\\
\implies\quad
(1-\epsilon/2) H_1 &~\leq~ H_\alpha
\\
\implies\quad
\frac{H_1}{H_\alpha} &~\leq~ \frac{1}{1-\epsilon/2} ~\leq~ 1 + \epsilon,
\end{align*}
the last inequality following from \Claim{exptaylor}.
This establishes the multiplicative approximation.

Let us now consider the above argument, replacing $\mu$ with
$\nu = \epsilon / ( 4 \log n \log m )$.
We obtain
$$
e(\alpha)
~\leq~ (\log m) \nu H_1
~\leq~ \epsilon / 4,
$$
since $H_1 \leq \log n$.
Thus, the additive approximation follows directly.
\end{proofof}

\subsection{Proofs from \Section{multipoint}}

Our first task is to prove \Lemma{tsallis-conv}.
We require a definition and two preliminary technical results.
For any integer $k \geq 0$ and real number $a \geq -1$, define
$$
G_k(a) ~=~ \sum_{i=1}^n x_i^{1+a} \, \ln^k(x_i),
$$
so $G_0(a) = F_{1+a} / ||A||_1^{1+a}$.
Note that $G_k^{(1)}(a) = G_{k+1}(a)$ for $k \geq 0$,
and $T(a) = (1 - G_0(a))/a$.

\begin{claim}
\ClaimName{tsallis-deriv}
The $k\th$ derivative of the Tsallis entropy has the following expression.
\[ T^{(k)}(a)
~=~
\frac{(-1)^k \, k! \, \big(1-G_0(a)\big)}{a^{k+1}}
\:-\:
\left(\sum_{j=1}^k
  \frac{(-1)^{k-j} \, k! \, G_j(a)}{a^{k-j+1}j!}\right)
\]
\end{claim}
\begin{proof}
The proof is by induction, the case $k=0$ being trivial.
So assume $k \geq 1$.
Taking the derivative of the expression for $T^{(k)}(a)$ above, we obtain:
\begin{align*}
& T^{(k+1)}(a)
\\
&=~ \left(\sum_{j=1}^k
  \frac{k!(k-j+1)(-1)^{(k+1)-j}G_j(a)}{a^{(k+1)-j+1}j!} +
  \frac{k!(-1)^{k-j}G_{j+1}(a)}{a^{k-j+1}j!}\right)
\\
&\oct +~ \frac{(-1)^{k+1}(k+1)!(G_0(a) - 1)}{a^{k+2}} +
\frac{(-1)^kk!G_1(a)}{a^{k+1}}
\\
&=~ \left(\sum_{j=1}^k
  \frac{k!(-1)^{(k+1)-j}G_j(a)}{a^{(k+1)-j+1}(j-1)!}\left(1 +
    \frac{k-j+1}{j}\right)\right) + \frac{G_{k+1}(a)}{a} +
\frac{(-1)^{k+1}(k+1)!(G_0(a) - 1)}{a^{k+2}}
\\
&=~ \left(\sum_{j=1}^{k+1}
  \frac{(k+1)!(-1)^{(k+1)-j}G_j(a)}{a^{(k+1)-j+1}j!}\right) +
\frac{(-1)^{k+1}(k+1)!(G_0(a) - 1)}{a^{k+2}}
\end{align*}
as claimed.
\end{proof}

\begin{claim}
\ClaimName{pre-hopital}
Define $S_k(a) = a^{k+1}T^{(k)}(a)$.  Then, for $1 \le j \le k + 1$,
\[
S_k^{(j)}(a) = \sum_{i=0}^{j-1} \binom{j-1}{i}
\frac{k!}{(k-j+i+1)!}a^{k-j+i+1} G_{k+1+i}(a)
\]
In particular, for $1 \le j \le k$, we have
\[ \lim_{a\rightarrow 0}S_k^{(j)}(a) = 0 
~~\quad\And\quad~~
\lim_{a\rightarrow 0}S_k^{(k+1)}(a) = k! \, G_{k+1}(0)
~~\quad\text{so that}\quad~~
\lim_{a\rightarrow 0}T^{(k)}(a) = \frac{G_{k+1}(0)}{k+1}.
\]
\end{claim}
\begin{proof}
We prove the claim by induction on $j$.  First, note
\[ S_k(a) = (-1)^k k!(1 - G_0(a)) - \left(\sum_{j=1}^k
  \frac{a^j(-1)^{k-j}k!G_j(a)}{j!} \right) \]
so that
\begin{eqnarray*}
S_k^{(1)}(a) &=& (-1)^{k-1}k!G_1(a) - \left(\sum_{j=1}^k
  -\frac{a^{(j + 1) - 1}(-1)^{k-(j+1)}k!G_{j+1}(a)}{((j + 1) - 1)!} +
  \frac{a^{j-1}(-1)^{k-j}k!G_j(a)}{(j-1)!}\right) \\
&=& a^kG_{k+1}(a)
\end{eqnarray*}
Thus, the base case holds.  For the inductive step with $2 \le j \le
k+1$, we have
\begin{eqnarray*}
S_k^{(j)}(a) &=&
\frac{\partial}{\partial
  a}\left(\sum_{i=0}^{j-2}\binom{j-2}{i}\frac{k!}{(k-j+i+2)!}a^{k-j+i+2}
    G_{k+1+i}(a)\right)\\
&=&
\sum_{i=0}^{j-2}\Bigg(\binom{j-2}{i}\frac{k!}{(k-j+i+1)!}a^{k-j+i+1}G_{k+1+i}(a)\\
&&{}+
\binom{j-2}{i}\frac{k!}{(k-j+(i+1)+1)!}a^{k-j+(i+1)+1}G_{k+1+(i+1)}(a)\Bigg)\\
&=&
\sum_{i=0}^{j-1}\binom{j-1}{i}\frac{k!}{(k-j+i+1)!}a^{k-j+i+1}G_{k+1+i}(a)\\
\end{eqnarray*}

The final equality holds since $\binom{j-2}{0} = \binom{j-1}{0} = 1$,
$\binom{j-2}{j-2} = \binom{j-1}{j-1} = 1$, and by Pascal's formula
$\binom{j-2}{i} + \binom{j-2}{i+1} = \binom{j-1}{i+1}$ for $0 \le i
\le j - 3$.

For $1 \le j \le k$, every term in the above sum is
well-defined for $a=0$ and contains a power of $a$ which is at least
$1$, so $\lim_{a\rightarrow 0}S_k^{(j)}(a) = 0$.  When $j = k+1$, all
terms but the first term contain a power of $a$ which is at least $1$,
and the first term is $k!G_{k+1}(a)$, so $\lim_{a\rightarrow
  0}S_k^{(k+1)}(a) = k!G_{k+1}(0)$.  The claim on $\lim_{a\rightarrow
  0} T{(k)}(a)$ thus follows by writing $T^{(k)}(a) = S_k(a)/a^{k+1}$
then applying l'H\^{o}pital's rule $k+1$ times.
\end{proof}

\vspace{6pt}
\begin{proofof}{\Lemma{tsallis-conv}}
We will first show that
\[
\left|T^{(k)}\left(-\frac{\epsilon}{(k+1)\ln m}\right)
- \frac{G_{k+1}(0)}{k+1}\right| \le \frac{6\epsilon \ln^k(m)H(x)}{k+1}
\]
Let $S_k(a) = a^{k+1}T^{(k)}(a)$ and note $T^{(k)}(a) =
S_k(a)/a^{k+1}$.  By \Claim{tsallis-deriv},
$\lim_{a\rightarrow 0} S_k(a) = 0$.  Furthermore, $\lim_{a\rightarrow
0} S_k^{(j)} = 0$ for all $1 \le j \le k$ by \Claim{pre-hopital}.
Thus, when analyzing $\lim_{a\rightarrow 0}
S_k^{(j)}(a) / (a^{k+1})^{(j)}$ for $0 \le j \le k$, both the
numerator and denominator approach $0$ and we can apply
l'H\^{o}pital's rule (here $(a^{k+1})^{(j)}$ denotes the $j$th
derivative of the function $a^{k+1}$).  By $k+1$ applications of
l'H\^{o}pital's rule, we can thus say that $T^{(k)}(a)$ converges to
its limit at least as quickly as $S_k^{(k+1)}(a) / (a^{k+1})^{(k+1)} =
S_k^{(k+1)}(a)/(k+1)!$ does (using \Claim{rudin}).
We note that
$G_j(a)$ is nonnegative for $j$ even and nonpositive otherwise.
Thus, for negative $a$, each term in the summand of the expression
for $S_k^{(k+1)}(a)$ in \Claim{pre-hopital} is nonnegative
for odd $k$ and nonpositive for even $k$.  As the analyses for even
and odd $k$ are nearly identical, we focus below on odd $k$, in
which case every term in the summand is nonnegative.  For odd $k$,
$S_k^{(k+2)}(a)$ is nonpositive so that $S_k^{(k+1)}(a)$ is
monotonically decreasing.  Thus, it suffices to show that
$S_k^{(k+1)}(-\epsilon/((k+1)\ln m))/(k+1)!$ is not much larger than
its limit.
  
\begin{eqnarray*}
\frac{S_k^{(k+1)}\left(-\frac{\epsilon}{(k+1)\ln m}\right)}{(k+1)!} &=&
\frac{\sum_{i=0}^{k} \binom{k}{i}
\frac{k!}{i!}\left(-\frac{\epsilon}{(k+1)\ln m}\right)^i
G_{k+1+i}\left(-\frac{\epsilon}{(k+1)\ln m}\right)}{(k+1)!}\\
&\le& \frac{1+2\epsilon}{k+1}\sum_{i=0}^{k} \binom{k}{i}
\left(\frac{\epsilon}{(k+1)\ln m}\right)^i |G_{k+1+i}(0)|\\
&\le& \frac{1+2\epsilon}{k+1}\sum_{i=0}^{k} k^i \left(\frac{\epsilon}{(k+1)\ln
    m}\right)^i |G_{k+1+i}(0)|\\
&\le& \frac{1+2\epsilon}{k+1}\sum_{i=0}^{k} \left(\frac{\epsilon}{\ln
    m}\right)^i |G_{k+1+i}(0)|\\
&\le& \frac{1+2\epsilon}{k+1}\sum_{i=0}^{k}\epsilon^i  |G_{k+1}(0)|\\
&\le& \frac{(1+2\epsilon)|G_{k+1}(0)|}{k+1} +
\frac{1+2\epsilon}{k+1}\sum_{i=1}^{k}\epsilon^i |G_{k+1}(0)|\\
&\le& \frac{(1+2\epsilon)|G_{k+1}(0)|}{k+1} +
\frac{2}{k+1}\sum_{i=1}^{k}\epsilon^i \ln^k(m)H(x)\\
&\le& \frac{|G_{k+1}(0)|}{k+1} + \frac{6\epsilon\ln^k(m)H(x)}{k+1}
\end{eqnarray*}

The first inequality holds since $x_i \ge 1/m$ for each $i$, so that
$x_i^{-\epsilon/((k+1)\ln m)} \le m^{\epsilon/((k+1)\ln m)} \le
m^{\epsilon/\ln m} \le e^{\epsilon} \le 1+2\epsilon$ for
$\epsilon \le 1/2$.  The final inequality above holds since
$\epsilon \le 1/2$.

The lemma follows since $|G_{k+1}(0)| \le \ln^{k}(m)H(x)$.
\end{proofof}

\begin{proofof}{\Lemma{eval-cheby}}
Let $P_j$ denote the $j\th$ Chebyshev polynomial.
We will prove for all $j \geq 1$ that
\[ P_{j-1}(x) \le P_j(x) \le P_{j-1}(x)\left(1+\frac{2j}{k^c}\right). \]
For the first inequality, we observe $P_{j-1} \in \mathcal{C}_j$,
so we apply \Fact{extremal-growth} together with
the fact that $P_j(y)$ is strictly positive for $y > 1$ for all $j$.

For the second inequality, we induct on $j$.
For the sake of the proof define $P_{-1}(x) = 1$
so that the inductive hypothesis holds at the
base case $d=0$.  For the inductive step with $j \ge 1$,
we use the recurrence definition of $P_j(x)$ and we have
\begin{eqnarray*}
P_{j+1}(x) &=& P_j(x) \left(1+\frac{2}{k^c}\right) + (P_j(x) - P_{j-1}(x)) \\
&\le& P_j(x)\left(1+\frac{2}{k^c}\right) + P_{j-1}(x)\frac{2j}{k^c}\\
&\le& P_j(x)\left(1+\frac{2}{k^c}\right) + P_j(x)\frac{2j}{k^c}\\
&=&   P_j(x)\left(1+\frac{2(j+1)}{k^c}\right)
\end{eqnarray*}
\end{proofof}

\comment{
\TODO{ Where does this go??} 

Now, consider the expression $T^{(k)}(a) =
(a^{k+1}T^{(k)}(a))/a^{k+1}$.  Both the numerator and denominator go
to zero as $a\rightarrow 0$.  Thus, we can apply L'H\^{o}pital's rule
to analyze this limit.  
Furthermore, \Claim{rudin} shows that by bounding the rate of convergence of
$(a^{k+1}T^{(k)}(a))'/((k+1)a^k)$ as $a\rightarrow 0$, we
obtain a bound for the rate of convergence of $T^{(k)}$ as well.  We
follow this strategy.
}

\subsection{Proofs from \Section{residual}}

\begin{fact}\FactName{reduce-gamma}
For any real $z>0$, $\Gamma(z+1) = z\Gamma(z)$.
\end{fact}

\begin{fact}\FactName{sine-bound}
For any real $z \ge 0$, $\sin(z) \le z$.
\end{fact}

\begin{fact}[Euler's Reflection Formula]\FactName{reflection}
For any real $z$, $\Gamma(z)\Gamma(1-z) = \pi/\sin(\pi z)$.
\end{fact}



\begin{definition} The function $V:\mathbb{R}^+\rightarrow \mathbb{R}$
  is defined by
\[  V(\alpha) = \frac{\left[\frac{2}{\pi}\Gamma(\frac{2\alpha}{3})
    \Gamma(\frac 13)
  \sin(\frac{\pi\alpha}{3})\right]^3}{\left[\frac{2}{\pi}
  \Gamma(\frac{\alpha}{3}) \Gamma(\frac 23)
  \sin(\frac{\pi\alpha}{6})\right]^6} - 1
\]
\end{definition}

\begin{lemma}
\LemmaName{limexists}
\[
\lim_{\alpha\rightarrow 0} V(\alpha) = \frac {\Gamma\left(\frac
  13\right)^3}{\Gamma\left(\frac 23\right)^6}
\]
\end{lemma}
\begin{proof}
Define $u(\alpha) = \Gamma(2\alpha / 3)(\pi\alpha/3) = \Gamma(2\alpha
/ 3)(2\alpha/3)(\pi / 2) = \Gamma((2\alpha/3) + 1)(\pi /2)$ by
\Fact{reduce-gamma}.  By the continuity of $\Gamma(\cdot)$ on
$\bR_+$, $\lim_{\alpha\rightarrow 0} u(\alpha) = \Gamma(1)\pi/2
= \pi/2$.  Define $f(\alpha) = \Gamma(2\alpha / 3)\sin(\pi\alpha/3)$.
Then $f(\alpha) \le u(\alpha)$ for all $\alpha\ge 0$ by
\Fact{sine-bound}, and thus $\lim_{\alpha\rightarrow 0} f(\alpha) \le
\pi/2$.  Now define $\ell_\delta(\alpha) = \Gamma(2\alpha /
3)(1-\delta)(\pi\alpha/3)$.  By the definition of the derivative and
the fact that the derivative of $\sin(\alpha)$ evaluated at $\alpha=1$
is $1$, it follows that $\forall\delta>0\ \exists\eps>0\hbox{ s.t. } 0
\le \alpha < \eps \Rightarrow \sin(\alpha) \ge (1 - \delta)\alpha$.
Thus, $\forall\delta>0\ \exists\eps>0\hbox{ s.t. } 0 \le \alpha < \eps
\Rightarrow \ell_{\delta}(\alpha) \le f(\alpha)$, and so $\forall
\delta > 0$ we have that $\lim_{\alpha\rightarrow 0} f(\alpha) \ge
\lim_{\alpha\rightarrow 0}\ell_{\delta}(\alpha) = (1 - \delta)\pi/2$.
Thus, $\lim_{\alpha\rightarrow 0} f(\alpha) \ge \pi/2$, implying
$\lim_{\alpha\rightarrow 0} f(\alpha) = \pi/2$.  Similarly we can
define $g(\alpha) = \Gamma(\alpha/3)\sin(\pi\alpha/6)$ and show
$\lim_{\alpha\rightarrow 0} g(\alpha) = \pi/2$.

Now,
\[ V(\alpha) = \frac{\left[\frac 2{\pi}\Gamma\left(\frac
      13\right)f(\alpha)\right]^3}{\left[\frac 2{\pi}\Gamma\left(\frac
      23\right)g(\alpha)\right]^6} \]
Thus $\lim_{\alpha\rightarrow 0}V(\alpha) = \Gamma(1/3)^3/\Gamma(2/3)^6$
as claimed.
\end{proof}

\begin{proofof}{\Lemma{stable-variance}}
Li shows in \cite{Li08b} that the variance of the geometric mean
estimator with $k=3$ is $V(\alpha)F_\alpha^2$.  As $\Gamma(z)$ and
$\sin(z)$ are continuous for $z \in \bR_+$, so is $V(\alpha)$.
Furthermore \Lemma{limexists} shows that $\lim_{\alpha\rightarrow 0} V(\alpha)$
exists (and equals $(\Gamma(1/3)^3/\Gamma(2/3)^6) - 1$).
We define $V(0)$ to be this limit.
Thus $V(\alpha)$ is continuous on $[0, 2]$, and the extreme value theorem implies there exists a
constant $C_{GM}$ such that $V(\alpha) \le C_{GM}$ on $[0,2]$.
\end{proofof}

\subsection{Detailed Analysis of Geometric Mean Residual Moments Algorithm}
\AppendixName{alg1}

Formally, define $R=\ceil{\log_{1+\frac{\eps}{c_1}} m}$,
and let $ I_z = \setst{i}{ (1+\frac{\eps}{c_1})^z \le |A_i| < (1+\frac{\eps}{c_1})^{z+1}}$
for $0 \leq z \leq R$.
Let $z^*$ satisfy $(1+\frac{\eps}{c_1})^{z^*} \le |A_1| < (1+\frac{\eps}{c_1})^{z^*+1}$.
For $1 \leq j \leq r$ and $0 \leq z \leq R$,
define $X_{j,z} = \sum_{i \in I_z} \mathbf{1}_{h_j(i) \neq h_j(1)}$.
We now analyze the $j\th$ trial.

\begin{claim}
\ClaimName{expect}
$\expect{ 2 \cdot F_{\alpha,j,1-h_j(1)} } = \big( 1 + O(\epsilon) \big) \cdot \Fres_\alpha$.
\end{claim}
\begin{proof}
We have
\begin{align*}
\expect{ 2 \cdot F_{\alpha,j,1-h_j(1)} }
 &~=~ 2 \cdot \expect{ \sum_i \abs{A_i}^\alpha \cdot \mathbf{1}_{h_j(i) \neq h_j(1)} } \\
 &~=~ 2 \cdot \sum_z \expect{ \sum_{i \in I_z}
    \abs{A_i}^\alpha \cdot \mathbf{1}_{h_j(i) \neq h_j(1)} } \\
 &~=~ 2 \cdot \sum_z \expect{ \sum_{i \in I_z} \big( (1 \pm \epsilon) (1+\epsilon)^z \big)^\alpha
    \cdot \mathbf{1}_{h_j(i) \neq h_j(1)} } \\
 &~=~ (1 \pm \epsilon)^\alpha \cdot \sum_z (1+\epsilon)^{z \alpha} \expect{ 2 X_{j,z} }.
\end{align*}
Clearly $\expect{ 2 \cdot X_{j,z} }$ is $\card{I_z}-1$ if $z=z^*$ and $\card{I_z}$ otherwise.
Thus
$$
\sum_z (1+\epsilon)^{z \alpha} \expect{ 2 \cdot X_{j,z} }
~=~ \sum_{i \geq 2} \big((1 \pm \epsilon) \abs{A_i} \big)^\alpha 
~=~ (1 \pm \epsilon)^\alpha \cdot \Fres_\alpha.
$$
Since $\alpha<2$, $(1 \pm \epsilon)^\alpha = 1 \pm O(\epsilon)$,
so this shows the desired result.
\end{proof}

We now show concentration for $X_z \defeq \frac 1r \sum_{1 \leq j \leq r} X_{j,z}$.
By independence of the $h_j$'s, Chernoff bounds show that
$X_z = (1 \pm \eps) \expect{ X_z }$ with probability at least
$1 - \exp(- \Theta(\eps^2 r) )$.
This quantity is at least $1-\frac{1}{8(R+1)}$ if we choose
$r = c_2 \ceil{\eps^{-2}(\log\log||A||_1 + \log(c_3/\eps))}$.
The \newterm{good event} is the event that, for all $z$, $X_z = (1\pm\eps) \expect{X_z}$;
a union bound shows that this occurs with probability at least $7/8$.
So suppose that the good event occurs.
Then a calculation analogous to \Claim{expect} shows that
\begin{align}
\nonumber
\sum_{j} \frac 2r \cdot F_{\alpha,j,1-h_j(1)}
 &~=~ (1 \pm \epsilon)^\alpha \cdot \sum_z (1+\epsilon)^{z \alpha} 
    \cdot 2 X_z \\
\nonumber
 &~=~ (1 \pm \epsilon)^\alpha \cdot \sum_z (1+\epsilon)^{z \alpha} 
    \cdot (1 \pm \epsilon) \expect{ 2 X_z } \\
\EquationName{Fconc}
 &~=~ \big(1 \pm O(\epsilon)\big) \cdot \Fres_\alpha.
\end{align}

Recall that $\tFres_\alpha = \sum_{j=1}^r \frac 2r \tF_{\alpha,j,1-h_j(1)}$.
Since the geometric mean estimator is unbiased, we also have that
\begin{equation}
\EquationName{tFres}
\expect{ \tFres_\alpha }
~=~ \expect{ \sum_j \frac 2r F_{\alpha,j,1-h_j(1)} }.
\end{equation}
We conclude the analysis by showing that the random variable
$\tFres_\alpha$ is concentrated.
By \Lemma{stable-variance} applied to each substream, and properties of variance, we have
$$
\var{ \tFres_\alpha }
~=~ \frac{4}{r^2} \sum_{j=1}^r \var{ \tF_{\alpha,j,1-h_j(1) } }
~\leq~ \frac{4 \, C_{GM}}{r} \cdot \expect{ \tF_{\alpha,j,1-h_j(1)} }^2
~\leq~ \frac{C_{GM}}{r} \cdot \expect{\tFres_\alpha}^2.
$$
Chebyshev's inequality therefore shows that
$$
\prob{~ \tFres_\alpha = (1\pm\eps) \expect{\tFres_\alpha} ~}
 ~\geq~ 1 - \frac{\var{\tFres_\alpha}}{(\eps \cdot \expect{\tFres_\alpha})^2}
 ~\geq~ 1 - \frac{C_{GM}}{\eps^2 \, r}
 ~>~ 6/7,
$$
by appropriate choice of constants.
This event and the good event both occur with probability at least $3/4$.
When this holds, we have
$$
\tFres_\alpha
 ~=~ (1\pm\eps) \expect{\tFres_\alpha}
 ~=~ (1\pm\eps) \expect{ \sum_j \frac 2r F_{\alpha,j,1-h_j(1)} }
 ~=~ \big(1 \pm O(\epsilon)\big) \cdot \Fres_\alpha,
$$
by \Equation{tFres} and \Equation{Fconc}.

\subsection{Proofs from \Section{res2}}
\AppendixName{res2}

\begin{proofof}{\Fact{good_hash}}
Let $B = \ceil{20/\epsilon}$ be the number of bins.
Let $\cH$ be a pairwise independent family of hash functions,
each function mapping $[n]$ to $[B]$.
Standard constructions yield such a family with $\card{\cH} = n^{O(1)}$.
We will let $h$ be a randomly chosen hash function from $\cH$.

For notational simplicity, suppose that $x_1 = \max_i x_i$.
Let $\cE_{i,j}$ be the indicator variable for the event that $h(i)=j$,
so that $\expect{ \cE_{i,j} } = 1/B$ and $\var{\cE_{i,j}} < 1/B$.
Let $X_j$ be the random variable denoting the weight of the items that
hash to bin $j$, i.e., $X_j = \sum_i x_i \cdot \cE_{i,j}$.
Since $\sum_i x_i = 1$, we have $\expect{ X_j } = 1/B$ and
$\var{ X_j } < \norm{x}_2^2 / B$.

Suppose that $x_1 \geq 1/2$.
Let $Y$ be the fraction of mass that hashes to $x_1$'s bin,
excluding $x_1$ itself.
That is, $Y = \sum_{i \geq 2} x_i \cdot \cE_{i,h(1)}$.
Note that $\expect{ Y } = ( \sum_{i \geq 2} x_i ) / B
< (\epsilon / 20) \cdot ( \sum_{i \geq 2} x_i )$.
By Markov's inequality,
$$
\prob{ Y \geq \epsilon \cdot ( \smallsum{i \geq 2}{} x_i ) }
~\leq~
\prob{ Y \geq 16 \expect{ Y } }
~\leq~ 1/16.
$$

Suppose that $x_1 < 1/2$.
This implies, by convexity, that $\norm{x}_2^2 < 1/2$.
Let $\beta = \sqrt{2/3} < 5/6$.
Then
$$
\prob{ \abs{ X_j - 1/B } \geq \beta } 
~\leq~ \frac{ \var{X_j} }{ \beta^2 }
~< \frac{3}{4B}.
$$
Thus, by a union bound,
$$
\prob{ \exists j \text{ such that } X_j \geq \beta + 1/B } ~\leq~
\frac{3}{4}.
$$

Suppose we want to test if $x_1 \geq 1/2$ by checking if there's a bin
of mass at least $5/6$. 
As argued above, the failure probability of one hash function is at
most $3/4$.  If we choose ten independent hash functions and check
that all of them have a bin of at least $5/6$, 
then the failure probability decreases to less than $1/16$.
\end{proofof}

\subsection{Proofs from \Section{renyitsallis}}
\AppendixName{renyitsallis}

\begin{proofof}{\Theorem{additive_renyi}}
Let $m_i$ be the number of times the $i$-th element appears in the
stream. Recall that $m$ is the length of the stream. By computing a
$(1+\eps')$-approximation to the $\alpha\th$ moment (as in \Fact
{stable_moment}) and dividing by $||A||_1^\alpha$, we get 
a multiplicative approximation to $F_\alpha/||A||_1^\alpha =
||x||_\alpha^\alpha$.
We can thus compute the value
$$\frac{1}{1-\alpha}\log\left((1\pm\eps')\sum_{i=1}^{n}x_i^\alpha\right) = 
\frac{1}{1-\alpha}\log\left(\sum_{i=1}^{n}x_i^\alpha\right)
+ \frac{\log(1\pm\eps')}{1-\alpha}= H_{\alpha}(X) \pm \frac{\eps'}{1-\alpha}.$$
Setting $\eps' = \eps \cdot |1-\alpha|$, we obtain
an additive approximation algorithm using 
$$
O\left(\left(\frac{|1-\alpha|}{\eps^2 \cdot |\alpha -1|^2} + \frac{1}{\eps \cdot |\alpha -1|}\right)\log m\right)
~=~ O(\log m/(|1-\alpha|\cdot\eps^2))
$$
bits, as claimed.
\end{proofof}

\posB
\begin{proofof}{\Theorem{additive_tsallis}}
If $\alpha \in (0,1)$, then because the function $x^\alpha$ is
concave, we get by Jensen's inequality
$$\sum_{i=1}^n{x_i}^{\alpha} \le n \cdot \left(\frac{1}{n}\right)^{\alpha}
= n^{1-\alpha}.$$
If we compute a multiplicative $(1+(1-\alpha)\cdot\eps\cdot
n^{\alpha-1})$-approximation to the $\alpha\th$ moment, we obtain an
additive $(1-\alpha)\cdot\eps$-approximation to $(\sum_{i=1}^n
x_i^\alpha)-1$. This in turn gives an additive $\eps$-approximation to
$T_\alpha$. By \Fact{stable_moment},
$$
O\left(\left(\frac{1-\alpha}{((1-\alpha)\cdot\eps\cdot
n^{\alpha-1})^2} + \frac{1}{(1-\alpha)\cdot\eps\cdot
n^{\alpha-1}}\right)\log m\right) 
~=~
O(n^{2(1-\alpha)}\log m/((1-\alpha)\eps^2))
$$
bits of space suffice to achieve the required approximation to the
$\alpha\th$ moment.

For $\alpha > 1$, the value $F_\alpha/||A||_1^\alpha$ is at most $1$,
so it suffices to approximate $F_\alpha$ to within a factor of
$1+(\alpha-1)\cdot\eps$. For $\alpha \in (1,2]$, again using
\Fact{stable_moment}, we can achieve this using
$O(\log m/((\alpha-1)\eps^2))$ bits
of space.
\end{proofof}

\posB
\begin{proofof}{\Lemma{large_difference}}
Consider first $\alpha \in (0,1)$. For $x \in (0,5/6]$,
$$\frac{x^{\alpha}}{x} = x^{\alpha-1} \ge
\left(\frac{5}{6}\right)^{\alpha-1}\ge 1 + C_1 \cdot (1-\alpha),$$
for some positive constant $C_1$. The last equality follows from
convexity of $(5/6)^y$ as a function of $y$.  Hence,
$$\sum_{i=1}^n x_i^\alpha \ge \sum_{i=1}^n (1 + C_1 (1-\alpha))x_i = 1
+ C_1 (1-\alpha),$$ and furthermore,
$$\left|1 - \sum_{i=1}^n x_i^\alpha\right| = \left(\sum_{i=1}^n
  x_i^\alpha\right) - 1 \ge C_1 \cdot (1-\alpha) = C_1 \cdot
|\alpha-1|$$

When $\alpha \in (1,2]$, then for $x \in (0,5/6]$,
$$\frac{x^{\alpha}}{x} = x^{\alpha-1} \le
\left(\frac{5}{6}\right)^{\alpha-1}\le 1 - C_2 \cdot (\alpha-1),$$
for some positive constant $C_2$. This implies that 
$$\sum_{i=1}^n x_i^\alpha \le \sum_{i=1}^n x_i (1 - C_2 \cdot
(\alpha-1)) = 1 - C_2 \cdot (\alpha-1),$$
and
$$\left|1 - \sum_{i=1}^n x_i^\alpha\right| = 1 - \sum_{i=1}^n
x_i^\alpha \ge C_2 \cdot (\alpha - 1) = C_2 \cdot |\alpha - 1|.$$

To finish the proof of the lemma, we set $C = \min\{C_1,C_2\}$.
\end{proofof}

\posB
\begin{proofof}{\Lemma{approximate_difference}}
We first argue that a multiplicative
approximation to $|1-x_i^\alpha|$ can be obtained from a multiplicative approximation to
$1-x_i$. Let $g(y) = 1 - (1-y)^{\alpha}$. Note that $g(1-x_i) = 1 - x_i^\alpha$. Since $1-x_i \in [0,1/3]$,
we restrict the domain of $g$ to $[0,1/3]$. The derivative of $g$ is $g'(y) = \alpha (1-y)^{\alpha - 1}$.
Note that $g$ is strictly increasing for $\alpha \in (0,1)\cup(1,2]$.
For $\alpha \in (0,1)$, the derivative is in the range $[\alpha,\frac{3}{2}\alpha]$.
For $\alpha \in (1,2]$, it always lies in the range $[\frac{2}{3}\alpha,\alpha]$.
In both cases, a $(1+\frac{2}{3}\eps)$-approximation to $y$ suffices to compute
a $(1+\eps)$-approximation to $g(y)$.

We now consider two cases:
\begin{itemize}
\item Assume first that $\alpha \in (0,1)$. For any $x \in (0,1/3]$, we have
$$\frac{x^{\alpha}}{x} \ge \left(\frac{1}{3}\right)^{\alpha-1} =
3^{1-\alpha} \ge 1 + C_1(1-\alpha),$$
for some positive constant $C_1$. The last inequality follows from the
convexity of the function $3^{1-\alpha}$.
This means that if $x_i < 1$, then
$$\frac{\sum_{j\ne i} x_j^\alpha}{1-x_i} \ge \frac{\sum_{j\ne i} x_j(1
  + C_1(1-\alpha))}{1-x_i}
= \frac{(1-x_i)(1 + C_1(1-\alpha))}{1-x_i} = 1 + C_1(1-\alpha).$$
Since $x_i \le x_i^{\alpha} < 1$, we also have 
$$\frac{\sum_{j\ne i} x_j^\alpha}{1-x_i^\alpha}\ge \frac{\sum_{j\ne i} x_j^\alpha}{1-x_i} \ge 1 + C_1(1-\alpha).$$
This implies that if we compute a multiplicative
$1+(1-\alpha)\eps/D_1$-approximations to both $1-x_i^{\alpha}$ and
$\sum_{j\ne i} x_j^\alpha$, for sufficiently large constant $D_1$, we
compute a multiplicative $(1+\eps)$-approximation of $(\sum_{j=1}^n
x_j^\alpha) - 1$.

\item The case of $\alpha \in (1,2]$ is similar.
For any $x \in (0,1/3]$, we have
$$\frac{x^{\alpha}}{x} \le \left(\frac{1}{3}\right)^{\alpha-1} \le 1 - C_2(\alpha-1),$$
for some positive constant $C_2$.
Hence, 
$$\frac{\sum_{j\ne i} x_j^\alpha}{1-x_i} \le \frac{\sum_{j\ne i} x_j(1 - C_2(\alpha-1))}{1-x_i}
= \frac{(1-x_i)(1 - C_2(\alpha-1))}{1-x_i} = 1 - C_2(\alpha-1),$$
and because $x_i^\alpha \le x_i$,
$$\frac{\sum_{j\ne i} x_j^\alpha}{1-x_i^\alpha} \le \frac{\sum_{j\ne i} x_j^\alpha}{1-x_i} \le 1 - C_2(\alpha - 1).$$
This implies that if we compute a multiplicative
$1+(\alpha-1)\eps/D_2$-approximations to both $1-x_i^{\alpha}$ and
$\sum_{j\ne i} x_j^\alpha$, for sufficiently large constant $D_2$, we
can compute a multiplicative $(1+\eps)$-approximation to
$1-\sum_{j=1}^n x_j^\alpha$.
\end{itemize}
\end{proofof}

\posB
\begin{proofof}{\Theorem{mult_tsallis}}
We run the algorithm of \Section{heavy} to
find out if there is a very heavy element.
This only requires $O(\log n)$ words of space. 

If there is no heavy element, then by \Lemma{large_difference}
there is a constant $C \in (0,1)$ such that
$|1-\sum_i{x_i^\alpha}| \ge C|\alpha-1|$. We want to compute a multiplicative
approximation to $|1-\sum_i{x_i^\alpha}|$. We know that the difference between
$\sum_i{x_i^\alpha}$ and 1 is large. Therefore, if we compute a multiplicative
$(1+\frac{1}{2}|\alpha-1|C\eps)$-approximation to $\sum_i{x_i^\alpha}$, we obtain
an additive $(\frac{1}{2}|\alpha-1|C\eps\sum_i{x_i^\alpha})$-approximation to $\sum_i{x_i^\alpha}$.
If $\sum_i{x_i^\alpha} \le 2$, then
$$\frac{\frac{1}{2}|\alpha-1|C\eps\sum_i{x_i^\alpha}}{|1-\sum_i{x_i^\alpha}|} \le \frac{|\alpha-1|C\eps}{C|\alpha-1|} = \eps.$$
If $\sum_i{x_i^\alpha} \ge 2$, then
$$\frac{\frac{1}{2}|\alpha-1|C\eps\sum_i{x_i^\alpha}}{|1-\sum_i{x_i^\alpha}|} \le
\frac{1}{2}|\alpha-1|C\eps \cdot 2 \le \eps.$$
In either case, we obtain a multiplicative $(1+\eps)$-approximation to $|1-\sum_i{x_i^\alpha}|$,
which in turn yields a multiplicative approximation to the Tsallis entropy.
We now need to bound the amount of space we use in this case. We use
the estimator of \Fact{stable_moment},
which uses $O(\log m/(|\alpha-1|\eps^2))$ bits in our case.
	
Let us focus now on the case when there is a heavy element. By
\Lemma{approximate_difference} it suffices to approximate
$\Fres_1$ and $\Fres_\alpha$, which we can do using
the algorithm of \Section{res2}.
The number of bits required is
$$O\left(\frac{\log m}{\eps \cdot |\alpha-1|}\right) + \tilde O\left(\frac{|\alpha -1| \cdot \log m}{(\eps \cdot |\alpha-1|)^2}\right) = \tilde O\left(\frac{\log m}{\eps^2 \cdot |\alpha-1|}\right).$$
\end{proofof}

\posB
\begin{proofof}{\Lemma{log_approx}}
For $t \in [4/9,1]$, the derivative of the logarithm function lies in the range $[a,b]$, where
$a$ and $b$ are constants such that $0 < a < b$. This implies that in this case,
a $(1+\eps)$-approximation to $t-1$ gives a $1+\frac{b}{a}\eps$ approximation to $\log(t)$.
We are given $y \in [1-t,(1+\eps)(1-t)]$, and we can assume that 
$y \in [1-t,\min\{5/9,(1+\eps)(1-t)\}]$.
We have
$$-\log(t) \le -\log(1-y),$$
and
\begin{eqnarray*}
\frac{-\log(1-y)}{-\log(t)} &\le& \frac{-\log(1-(1+\eps)(1-t))}{-\log(t)} = \frac{-\log(t - \eps(1-t))}{-\log(t)}\\
&\le& \frac{-\log(t) + (-\log(t - \eps(1-t)) + \log(t))}{-\log(t)}\\
&\le& 1 + \frac{-\log(t - \eps(1-t)) + \log(t)}{-\log(t)}\\
&\le& 1 + \frac{\eps(1-t)\cdot\max_{z\in[\max\{t - \eps(1-t),4/9\},t]}(\log(z))'}{(1-t)\cdot\min_{z\in[4/9,1]}(\log(z))'}\\
&\le& 1 + \frac{\eps(1-t)\cdot\max_{z\in[t,1]}(\log(z))'}{(1-t)\cdot\min_{z\in[4/9,1]}(\log(z))'}\\
&\le& 1 + \frac{\eps(1-t)\cdot b}{(1-t)\cdot a} = 1 + \frac{b}{a}\eps.
\end{eqnarray*}
Consider now $t > 1$. We are given $y \in [t-1,(1+\eps)(t-1)]$,
and we have
$$\log(t) \le \log (y+1) \le \log((1+\eps)(t-1)+1).$$
Furthermore, 
\begin{eqnarray*}
\frac{\log((1+\eps)(t-1)+1)}{\log(t)} &\le& \frac{\log(t) + \log((1+\eps)(t-1)+1) - \log (t)}{\log(t)} \\
&=& 1 + \frac{\log(t+(t-1)\eps)-\log(t)}{\log(t)}\\
&=& 1 + \frac{\int_{t}^{t+(t-1)\eps}(\log(z))'dz}{\int_{1}^t(\log(z))'dz}\\
&\le& 1 + \frac{(t-1)\eps\max_{z\in[t,t+(t-1)\eps]}(\log(z))'}{(t-1)\max_{z\in[1,t]}(\log(z))'}\\
&\le& 1 + \frac{(t-1)\eps}{t-1} = 1 + \eps.
\end{eqnarray*}
Hence, we get a good multiplicative approximation to $\log(t)$.
\end{proofof}

\posB
\begin{proofof}{\Theorem{mult_renyi}}
We use the algorithm of \Section{heavy} to
check if there is a single element of high frequency.
This only requires $O(\log m)$ bits of space.

If there is no element of frequency greater than $5/6$, then the R\'enyi entropy for any $\alpha$ is greater than 
the min-entropy $H_\infty = -\log \max_i x_i \ge \log (6/5)$. Therefore, in this case it suffices to run the additive approximation algorithm with $\eps' = \log (6/5) \eps$ to obtain a sufficiently good estimate. To run that algorithm, we use $O\left(\frac{\log m}{|1-\alpha|\eps^2}\right)$ bits of space.

Let us consider the other case, when there is an element of frequency
at least $2/3$. For $\alpha\in(1,2]$, we have
$$\left(\frac{2}{3}\right)^2 \le \sum_i x_i^\alpha \le 1,$$
and for $\alpha \in (0,1)$, $\sum_{i=1}^n x_i^\alpha \ge 1$.
Therefore, by \Lemma{log_approx}, it suffices to compute a
multiplicative approximation to $|1-\sum_i x_i^\alpha|$, which we can
do by \Lemma{approximate_difference}. By
algorithms from \Section{res1} and \Section{res2}, we can compute the multiplicative
$(1+\Theta(|1-\alpha|\eps))$-approximations required by
\Lemma{approximate_difference} with the same space
complexity as for the approximation of Tsallis entropy (see the proof
of Theorem~\ref{thm:mult_tsallis}).
\end{proofof}

\posB
\begin{proofof}{\Theorem{lastone}}
\comment{We reduce from the communication game of one-way $t$-party
disjointness, which has a lower bound of $\Omega(n/t^{1+\gamma})$ bits
of communication against randomized communication protocols for any
$\gamma > 0$ \cite{BJKS04}.  In this problem there are $t$
players, where player $i$ has a string $x_i\in\{0,1\}^n$.  The string
$x_i$ should be interpreted as the characteristic vector of a set
$S_i\subseteq [n]$.  The players are guaranteed that either
$S_i\cap S_j = \emptyset$ for all $i\neq j$ (``yes'' instance), or
that there exists some $z\in [n]$ such that $S_i\cap S_j = \{z\}$ for
all $i\neq j$ (``no'' instance).  We measure the number of bits of
communication of a protocol by the total number of bits sent by all
players (and not by the maximum communication of a single player).
}
The proof is nearly identical to that of Theorem~3.1 in \cite{BJKS04}.
We need merely observe that if ${\tilde{H}}_{\alpha}$ is a
$(1+\varepsilon)$-approximation to $H_{\alpha}$, then
$m^{\alpha(1+\varepsilon)}2^{(1-\alpha){\tilde{H}}_{\alpha}}$ is a
multiplicative $m^{\alpha\varepsilon}$-approximation to
$F_{\alpha}$.  From here, we set $t = cm^{\varepsilon}n^{1/\alpha}$
and argue identically as in \cite{BJKS04} via a reduction from
$t$-party disjointness; we omit the details.
\end{proofof}

\end{document}